\documentclass[11pt,aps]{article}

\usepackage{graphicx}              
\usepackage[fleqn]{amsmath}   
\usepackage{mathtools}           
\usepackage{amsfonts}              
\usepackage{amsthm}  
\usepackage{lineno}
\usepackage{color}

\numberwithin{equation}{section}
\renewcommand\theequation{\arabic{section}.\arabic{equation}}

\usepackage{amsfonts,amssymb,amsmath,epsfig,epic,color,amsthm}
\usepackage{dcolumn}
\usepackage{bm}

 \usepackage[titletoc, title]{appendix}

\theoremstyle{plain}

\newtheorem{thm}{Theorem}
\newtheorem{lem}{Lemma}
\newtheorem{prop}{Proposition}
\newtheorem{cor}{Corollary}

\theoremstyle{definition}

\newtheorem{conj}{Conjecture}

\theoremstyle{remark}
\newtheorem{rem}{Remark}

\newcommand{\beqs}{\begin{eqnarray}}
\newcommand{\eeqs}{\end{eqnarray}}

\DeclareMathOperator{\Ric}{Ric}

\DeclareMathOperator{\Tr}{Tr}

\DeclareMathOperator\erfc{erfc}

\begin{document}

\title{On the Number of Bound States of Point Interactions on Hyperbolic Manifolds}

\author{{Fatih ERMAN}
\\ \and
{Department of Mathematics, \.{I}zmir Institute of Technology, Urla,
35430, Izmir, Turkey}
\\
{e-mail: fatih.erman@gmail.com}}

\maketitle

\begin{abstract}
We study the bound state problem for $N$ attractive point Dirac $\delta$-interactions in two and three dimensional Riemannian manifolds. We give a sufficient condition for the Hamiltonian to have $N$ bound states and give an explicit criterion for it in hyperbolic manifolds $\mathbb{H}^2$ and $\mathbb{H}^3$. Furthermore, we study the same spectral problem for a relativistic extension of the model on $\mathbb{R}^2$ and $\mathbb{H}^2$.
\end{abstract}

\begin{flushleft}

 \end{flushleft}

\section{Introduction} \label{Introduction}

Point interactions were first introduced as a solvable toy model describing the short range interactions in nuclear and solid state physics.  There exists a vast amount of literature about them from several perspectives. The reader is invited to consult the books \cite{Albeverio 1, Albeverio 2, Demkov} and references therein for a more detailed study. The formal Hamiltonian is given by
\beqs
H=- {\hbar^2 \over 2m} \Delta  - \sum_{i=1}^{N} \lambda_i \delta(\mathbf{x} - \mathbf{a}_i) \;, \label{Hamiltonianflat}
\eeqs
where $\Delta$ is the Laplacian{\color{red},} and $\lambda_i$'s are the coupling constants (also called strength or intensity of the potential), which are assumed to be positive for all $i=1,2,\ldots, N$ and $\mathbf{a}_i$'s are the locations of the Dirac-$\delta$ centers in $\mathbb{R}^D$. One reason why the subject attracts a great deal of interest is that the point interactions (or Dirac-$\delta$ potentials) in two and three dimensions requires renormalization procedure. Moreover, the formal Hamiltonian describing them was not a well-defined self-adjoint operator in Sobolev spaces $W^{2,2}(M)$ so that one must clarify the meaning of the formal Hamiltonian. 
In order to accomplish this, one should develop a mathematically rigorous way which corresponds to the intuitive notion of Dirac-$\delta$ potential. One possible approach is to construct rigorously an operator associated with the formal Hamiltonian (\ref{Hamiltonianflat}) through the self-adjoint extension theory of symmetric operators \cite{Albeverio 1}. 
Historically, 
the first rigorous approach to the problem in $\mathbb{R}^3$ was given by Berezin and Faddeev \cite{Berezin Faddeev} and summarized in section 1.5 of \cite{Albeverio 2}. In there, the Hamiltonian is first approximated by the sequence of operators using  the spectral  representation of the Laplace operator. By choosing the sequence of functions converging to the Dirac delta function through the Fourier transformation, it is then possible to calculate the resolvent of the sequence of the operators explicitly and to show that it has a nontrivial limit if and only if a sequence for coupling constants is chosen properly (coupling constant renormalization).  Alternatively, another approach has been discussed in chapter 2.1 of \cite{Albeverio 1}: The formal Hamiltonian is first treated as a singular perturbation of the free Hamiltonian. Then, the Fourier transform of that ill-defined formal Hamiltonian with a momentum cut-off becomes a finite rank perturbation of the free Hamiltonian, and the coupling constant is chosen as a function of the cut-off in such a way that the resolvent of the regularized Hamiltonian in Fourier space has a non-trivial limit as the cut-off is removed. In both approaches, after choosing the coupling constant, the sequence of the self-adjoint operators converges to a self-adjoint operator in the strong resolvent or norm resolvent sense. Apart from the self-adjoint extension approaches developed by von Neumann and Krein, and the above two approximation procedures, there are other approaches to point interactions, namely non-standard analysis and the theory of quadratic forms \cite{Albeverio 1}.

Point interactions have been generalized onto some particular surfaces in $\mathbb{R}^3$ (onto infinite planar strip as a natural model for quantum wires containing impurities and onto torus using the Von Neumann's and Krein's theory of self-adjoint extensions (see \cite{Rudnick torus, Exner strip} and references therein))  and their spectral properties have been studied in great detail. Moreover, they have been constructed rigorously on even more general spaces, e.g. Riemannian manifolds of bounded geometry in \cite{BGP}. Our approach in this work is to study some interesting problems for the bound state spectrum of the point interactions on some class of Riemannian manifolds. 
In order to keep the present paper self-contained and make the reading of this paper easier, we recall some of the results basically established in our previous works constructed on compact manifolds
and Cartan-Hadamard manifolds with Ricci tensor bounded below (by which we mean $\Ric(.,.) \geq  c \; g(.,.) $)  in two and three dimensions \cite{point interactions on manifolds1, point interactions on manifolds2, caglar} (they were not stated as theorems in there) through Theorem \ref{theorem1} and Theorem \ref{compact spectrum}. Their proofs are given in Appendices. 
In there, we basically follow a strategy similar to the above first approximation procedure by using the heat kernel. In contrast to the flat case, Fourier transformation was useless since it cannot be defined globally on a generic Riemannian manifold. 
After reviewing our previous results, we show in this work that the principal matrix given in the resolvent formula is a matrix-valued holomorphic function on the complex plane, where $\Re(z)<0$ for the above-mentioned manifolds. This is done by using the explicit closed expression of the principal matrix without going into the theory of Nevanlinna functions. Hence,
we also justify some a priori assumptions in our previous works and improve our earlier somewhat heuristic calculations.

The estimates for the number of bound states of a Schr\"{o}dinger operator for a particular class of potentials is extensively discussed in \cite{Reed Simon V4}. For $N$ point $\delta$-interactions, it is well-known that there exist at most $N$ bound states in flat spaces  \cite{Albeverio 1, Albeverio 2}. Moreover, necessary and sufficient condition for the one dimensional Schr\"{o}dinger operator with finitely many point $\delta$-interactions to have the same number of negative eigenvalues as the number of point interactions  is given in \cite{Albeverio Nizhnik 1} and  
an effective algorithm for determining the number of negative eigenvalues  is constructed in \cite{Albeverio Nizhnik 2}. It has been proved in \cite{Ogurisu 1, Ogurisu 2}  that the number of negative eigenvalues is less than or equal to the number of negative coupling constants and  necessary and sufficient conditions are given for it to satisfy the saturation value of the bound.     
In \cite{Goloshchapova Oridoroga 1, Goloshchapova Oridoroga 2}, the number of negative eigenvalues is shown to be equal to the number of negative eigenvalues of a certain class of finite Jacobi matrices and given independently a necessary and sufficient condition for the same problem to satisfy the above saturated bound. Multi-dimensional extension of the above results has been carried out in a recent work \cite{Ogurisu 3}.   
The main aim of this work is to give a sufficient condition for the Hamiltonian (after the renormalization procedure) to satisfy that the number of bound states equals to the number of point $\delta$-interactions and determine an explicit criterion for that in hyperbolic spaces $\mathbb{H}^2$ and $\mathbb{H}^3$. 
Our proof is essentially the extension of the work \cite{Ogurisu 1} to the curved spaces, namely to the hyperbolic manifolds. However, the matrix elements in the resolvent formula here are closed analytic expressions in terms of the heat kernel for a generic Riemannian manifold in contrast to the explicit analytic formula in flat spaces. 
Finally, the same spectral problem is discussed for a relativistic version of the model on $\mathbb{R}^2$ and $\mathbb{H}^2$.

\textbf{Notation.} The notation in this work is slightly  different from the one usually used in mathematics literature \cite{Albeverio 1}. We also use some terminology from quantum field theory since the point interactions in two and threee dimensional quantum mechanics are considered as a toy model for understanding many concepts originally introduced in quantum field theory, e.g., regularization, dimensional transmutation, renormalization, renormalization group, asymptotic freedom, etc. (see \cite{Jackiw} and also a recent work \cite{AlHashimi}). The notation used here can easily be converted to the one used in \cite{Albeverio 1}. For instance, the principal matrix $\Phi$ introduced here is exactly the matrix $\Gamma$ in \cite{Albeverio 1} up to a unitary transformation and all the others are implicitly related.


\section{Point Interactions on Riemannian Manifolds}
\label{Point Interactions on Riemannian Manifolds}

We consider a single quantum mechanical particle intrinsically moving in a $D$ - dimensional Riemannian manifold $M$ with the metric structure $g$ (that is, the particle is constrained to $M$ a priori) in the presence of finitely many point $\delta$-interactions.  In this approach, an ordering ambiguity arises, and it leads to multiple quantization procedures which differ by a term proportional to the scalar Ricci curvature in the Hamiltonian \cite{DeWitt}. If one is interested in the class of manifolds with constant scalar curvature, then the effect of this term is simply a shift in the spectrum of the Hamiltonian. Here, we are not taking into account this curvature term for simplicity since one can essentially construct the model with this additional term. For this reason, we assume that the free Hamiltonian of a particle in $(M,g)$ is chosen as
\beqs
H_0 = - \Delta_g \;,
\eeqs
where $\Delta_g = \frac{1}{\sqrt{\mathrm{det}(g)}}
\sum_{i,j=1}^{D} \frac{\partial}{\partial x^i}
\left(g^{i j} \, \sqrt{\mathrm{det} (g)} \;
\frac{\partial}{\partial x^j}\right)$ is the Laplace-Beltrami operator (or simply Laplacian) written in local coordinates $\{x^i \}$ on a $D$-dimensional Riemannian manifold $(M,g)$. We use the units such that $\hbar=2m=1$.

Then, Hamiltonian for a single particle moving in $M$ and interacting with attractive point interactions $\delta_{g,a_i}$ supported by a finite set of isolated points $a_i \in M$ is formally given by
\beqs H=- \Delta_g  - \sum_{i=1}^{N} \lambda_i \delta_{g,a_i}(.) \;, \label{Hamiltonian}
\eeqs
where $\delta_{g,a_i}$ in the interaction term denotes the point-like Dirac $\delta$-function supported by the points $a_i \in M$ (it is defined as a continuous linear functional acting on the space of test functions $f(x)$ on $M$: $\delta_{g,a_i}(f)=f(a_i)$, or sometimes formally written as $\int_{\mathcal{M}} \delta_{g}(x,a_i) \; f(x) d_{g}^{D}x = f(a_i)$). Moreover, we suppose that $a_i \neq a_j$ for $i \neq j$. 

\textit{Unless otherwise stated throughout the paper, we restrict $(M,g)$ to two and three -dimensional Riemannian manifolds without boundary and consider two important classes of Riemannian manifolds, namely compact Riemannian manifolds and Cartan-Hadamard manifolds (geodesically complete, simply connected, noncompact Riemannian manifolds with nonpositive sectional curvature everywhere) with Ricci tensor bounded from below ($Ric(.,.) \geq  c \; g(.,.) $).}

\textit{We call the point spectrum below the spectrum of the free Hamiltonian as bound state spectrum.} Since we can always shift the spectrum by a constant without altering physics, we will always assume that the bound state spectrum lies on the negative real axis for the above class of manifolds. We now recall the essential part of the construction of the model which was already established in \cite{point interactions on manifolds2, caglar}. We here state them as a theorem and shortly give its proof in Appendix \ref{proofoffirsttheorem} for the sake of completeness of our paper.

\newtheorem{thm1}{Theorem}
\begin{thm} \label{theorem1}
Let $M$ be a compact Riemannian manifold without boundary with Ricci curvature bounded from below ($Ric(.,.) \geq  c \; g(.,.)$) or a Cartan-Hadamard (C-H) manifold without boundary, and  $H_{\epsilon}$ be the self-adjoint operator in $L^2(M)$, given by
\beqs H_\epsilon \psi(x)=-\Delta_g \psi(x)  - \sum_{j=1}^N \lambda_j(\epsilon) K_{\epsilon}(x,a_j) \int_{M} K_{\epsilon}(y,a_j) \psi(y) \; d^{D}_{g}y  \;, \label{regularized H}\eeqs 
where
$K_{\epsilon}(x,y)$ is the heat kernel defined on $M$ and $d^{D}_{g}x=\sqrt{\det(g)} dx^1 \ldots dx^D$ is the Riemannian volume form in the local coordinates. If the coupling constants $\lambda_i$'s are chosen as
\beqs \label{barecouplingdeltamanifolds} {1 \over
\lambda_{i}(\epsilon)} =  \int_\epsilon^\infty  K_t(a_i,a_i) \; e^{-t
    \mu_i^2}  \; d t \;, 
\eeqs
with $\mu_i>0$ (from the renormalization point of view, $-\mu_{i}^{2}$ is the experimentally measured bound state energy of the particle to the $i$-th point interaction while all the other centers are sufficiently far away from the $i$-th one), then for $\Re(z)<0$ sufficiently large, the resolvent of the regularized Hamiltonian (\ref{regularized H}) as $\epsilon \rightarrow 0$ converges to the following nontrivial limit (known as Krein's resolvent formula)
\beqs \label{resolvent} R(z) f(x) = R_0(z) \; f(x) + \sum_{i,j=1}^N
R_0(x,a_i|z)\, \left[\Phi^{-1}(z)\right]_{ij} R_0(z) \; f(a_j)\;, \eeqs
where $R_0(z) \; f(x)= (-\triangle_g - z)^{-1} f(x)= \int_{M} R_0(x,y|z) \; f(y) \; d^{D}_{g} y $ and
\beqs \label{phiheat3} \Phi_{ij} (z) =
\begin{cases}
\begin{split}
\int_0 ^\infty  K_{t}(a_i,a_i) \;
\left(e^{-t\mu_i^2} -e^{t z}\right) \; d t
\end{split}
& \mathrm{if} \;\;  i = j \\[2ex]
\begin{split}
- \int_0^\infty K_{t}(a_i,a_j)
\; e^{t z} \; d t
\end{split}
& \mathrm{if} \;\; i \neq j.
\end{cases} \;,
\eeqs
called the principal matrix and $R_0(x,y|z)=\int_{0}^{\infty} e^{z t} K_t(x,y) \; d t$ is the free resolvent kernel. Moreover, there exists a unique self-adjoint operator, say $H$, associated with the resolvent (\ref{resolvent}). Hence, the operator $H_\epsilon$ converges to the self-adjoint operator $H$ in the strong resolvent sense.
\end{thm}

\begin{rem}
The motivation for choosing the coupling constants (\ref{barecouplingdeltamanifolds})
is due to the short time asymptotic expansion of diagonal heat
kernel
\beqs K_t(x,x) \sim  {1 \over (4 \pi  t)^{D/2}} \sum_{k=0}^{\infty} u_k(x,x) \; t^k  \;, \label{asymheat}
\eeqs
for any point $x$ in a $D$-dimensional Riemannian manifold without boundary \cite{Gilkey}. Here $u_k(x, x)$ are scalar
polynomials in the curvature tensor of the manifold and its covariant derivatives at point $x$.
\end{rem}

\begin{rem} \label{heatkernelbounds}
Note that all the matrix elements of the principal matrix $\Phi$ are bounded for $\Re(z)<0$ thanks to the exponentially damping terms in the upper bounds of the heat kernel related to the geometry of $M$. In particular, based on the estimate given in \cite{Wang}, the upper bound of the heat kernel for compact manifolds with Ricci curvature bounded from below ($Ric(.,.) \geq  c \; g(.,.)$) \cite{point interactions on manifolds2}, is given by
\beqs \label{heatkernel bound1} K_t(x,y) \leq
\left[{C_1 \over V(M)} + {C_2 \over t^{D/2}}
\right] \; \mbox{exp} \left(-\frac{d^{2}(x,y)}{C_3 t}\right) \;,
\eeqs
for all $x,y \in M$ and $t>0$. For Cartan-Hadamard (C-H) manifolds \cite{Grigoryan Heat Book, Grigoryan 2ndbook}, one has
\beqs \label{heatkernel bound2} K_t(x,y) \leq
{C_4 \over t^{D/2}}  \; \mbox{exp} \left(-\frac{d^{2}(x,y)}{C_5
t}\right) \;,
\eeqs
for all $x,y \in M$ and $t>0$. Here $V(M)$ is the volume of the manifold and $d(x,y)$
is the geodesic distance between the point $x$ and $y$ on $M$.  All the constants $C_1, C_2, \cdots$ are dimensionless and depend on the geometry. In particular, the constants $C_3$ and $C_5$ are strictly greater than $4$.
\end{rem}

\begin{lem} \label{CheegerYau} [Cheeger - Yau \cite{CheegerYau}] 
If the Riemannian manifold is complete and has a Ricci tensor
bounded from below, i.e., $\Ric(.,.) \geq - (D-1) k \; g(.,.) $, with $k \in \mathbb{R}$, then we have
the following lower bound for the heat kernel:
\beqs 
K_{t} (x, y) \geq K_{t}^{k} (d (x, y)) \;, \label{heatkernellowerbound}
\eeqs
where $K_{t}^{k}$ is the heat kernel of the simply connected complete manifold of
constant sectional curvature $k$. 
\end{lem}

\begin{rem} \label{Cheeger-Yau remark}
In particular, we may choose $K_{t}^{k}(d(x, y))$ as the heat kernel on the hyperbolic manifolds $\mathbb{H}^D_{\kappa}$ of constant negative sectional curvature $-\kappa^2$ since they are explicitly known \cite{Grigoryan Heat Book}:
\beqs \label{lowerboundheatkernelhperbolicspaces} 
K_{t}^{\kappa}(d(x,y)) =
\begin{cases}
\begin{split}
{\sqrt{2} \over \kappa} {1 \over (4 \pi t)^{3/2}} \; e^{- \kappa^2 t/4} \int_{ \kappa  d(x,y)}^{\infty}  
{s \; e^{-s^2/4 \kappa^2 t} \over \sqrt{\cosh s -\cosh \kappa d(x,y)}} \; d s  \;,
\end{split}
& \mathrm{for} \; D=2 \\[2ex]
\begin{split}
 { \kappa d(x,y)  \over (4 \pi t)^{3/2} \sinh  \kappa d(x,y)} \; e^{- \kappa^2 t -{d(x,y)^2 \over 4 t}} \;,
\end{split}
& \mathrm{for} \; D=3 \;.
\end{cases}
\eeqs
In case the lower bound for the Ricci curvature is positive, we may choose the lower bound as the heat kernel on $D$-dimensional flat space and the argument below becomes even simpler.
\end{rem}

\begin{lem} \label{holomorphic}
The principal matrix $\Phi(z)$ for compact manifolds with Ricci tensor bounded from below ($\Ric(.,.) \geq  c \; g(.,.) $) and for Cartan-Hadamard manifolds is a matrix-valued holomorphic  function on the complex plane, where $\Re(z)<0$. In particular, it has a branch cut along $[(D-1)^2 \kappa^2/4, \infty)$ for $D$-dimensional hyperbolic spaces of sectional curvature $-\kappa^2$. 
\end{lem}

\begin{proof}
The proof is essentially based on the following theorem (theorem 1.1 in Chapter 2 of \cite{Olver}): Let $t$ be a real variable ranging over the interval $(0,\infty)$ and $z$ a complex variable ranging over a domain $\mathcal{R}$. Assume that the function $f(z,t)$ satisfies the following conditions: (i) $f(z,t)$ is a continuous function of both variables. (ii) For each fixed value of $t$, $f(z,t)$ is a holomorphic function of $z$. (iii) The integral $F(z)=\int_{0}^{\infty} f(z,t)\; d t$ converges uniformly at both limits in any compact set in $\mathcal{R}$. Then, $F(z)$ is holomorphic in $\mathcal{R}$ and its derivatives of all orders may be found by differentiating under the integral sign.  
It is self-evident that two hypotheses of the above theorem applied to the matrix elements of the principal matrix $\Phi$ are satisfied  
since the heat kernel $K_t(x,y)$ defined on $M \times M \times (0,\infty)$ is $C^1$ - function with respect to the variable $t$ and exponential function $e^{t z}$ is an entire function for each fixed value of $t$.   What is left is to show that all the matrix elements converge uniformly on a compact subset of the chosen region $\mathcal{R}$. Let $\mathcal{R}$ be the complex plane with $\Re(z)<0$.
Here we choose the compact subset of the region as $\mathcal{D}=\{ z \in \mathbb{C}| - \epsilon_2 \leq \Re(z) \leq -\epsilon_1\; \& \; \eta_2 \leq \Im(z) \leq \eta_1 \}$, where $\epsilon_{1}, \epsilon_2$ are positive. 
We first prove the uniform convergence for the diagonal part of the principal matrix on $\mathcal{D}$. Since the integrand is unbounded due to the short time asymptotic expansion of the diagonal heat kernel (\ref{asymheat}), we split the integral into two parts: $\int_{0}^{1}  K_{t}(a_i,a_i) \;
(e^{-t\mu_i^2} -e^{t z}) \; d t$ and $\int_1 ^\infty  K_{t}(a_i,a_i) \;
(e^{-t\mu_i^2} -e^{t z}) \; d t$. We first use the upper bounds of the heat kernel for compact manifolds with Ricci tensor bounded below ($\Ric(.,.) \geq  c \; g(.,.)$) and for Cartan-Hadamard manifolds given in the remark \ref{heatkernelbounds}. For the sake of simplicity, we do not have to analyse the problem separately for each class of manifold since the volume term in the upper bound can be combined into the proof by essentially following the same line of arguments. In the first integral, we have
\beqs
|K_t(a_i,a_i) \; (e^{-t\mu_i^2} -e^{t z}) | & \leq & C_4 \left|{e^{-t\mu_i^2} -e^{t z} \over t^{D/2}}\right| \;,
\eeqs 
for all $t>0$ and $i=1, \ldots, N$. If we define the following holomorphic function  $f(z)=-{e^{t z} \over t^{D/2}}$ for each value of $t>0$, then it is easy to show that $|f(z)-f(-\mu_{i}^2)| =|\int_{\gamma} f'(\zeta) d\zeta |\leq \max_{\zeta \in \mathcal{D}} |f'(\zeta)| L(\gamma)$ for any curve $\gamma$ connecting $-\mu_{i}^2$ to any $z$ in the above compact region $\mathcal{D}$. Then, we can always choose $\gamma$ as a straight line on $\mathcal{D}$ connecting the points $-\mu_{i}^2$ and $z$, i.e., $L(\gamma)=|z+\mu_{i}^2|$. Hence we obtain 
\beqs
|K_t(a_i,a_i) \; (e^{-t\mu_i^2} -e^{t z}) | & \leq & C_4 |z+\mu_{i}^2| \max_{\zeta \in D} { e^{t \Re(\zeta) } \over t^{{D \over 2}-1}} \leq  C_4 (\sqrt{\epsilon_{1}^2 +\eta_{1}^2} +\mu_{i}^2)\; { e^{- t \epsilon_1} \over t^{{D \over 2}-1}} \;,
\eeqs 
and the right hand side of the inequality is integrable on the interval $(0,1)$ for $D=2,3$. In the second integral, we have $|K_t(a_i,a_i) \; (e^{-t\mu_i^2} -e^{t z}) | \leq C_4 |e^{-t\mu_i^2} -e^{t z} | \leq C_4 (e^{-t \mu_{i}^2} + e^{-t \epsilon_1}) $, and this is clearly integrable on $(1,\infty)$. As for the off-diagonal matrix elements of the principal matrix, we have $|K_t(a_i,a_j) e^{t z}| \leq C_4 \exp (- d^2(a_i,a_j)/ C_3 t) t^{-D/2}$ in the region $\mathcal{D}$, which is integrable on $(0,\infty)$. Hence, we show that all the matrix elements of the principal matrix are uniformly convergent on the compact subset $\mathcal{D}$ of $\mathcal{R}$ as a consequence of the Weierstrass's M-test \cite{Copson}. Since all its matrix elements of $\Phi$ are holomorphic, the principal matrix $\Phi$ is matrix-valued holomorphic function on $\mathcal{R}$  and 
the derivatives of all orders of $\Phi$ with respect to $z$ can be found by differentiating under the sign of integration.

If we do not know the exact explicit expression of the principal matrix, it is in general difficult and rather involved to determine its branch cut structure for a generic class of Riemannian manifold. 
For the three dimensional hyperbolic spaces $\mathbb{H}^{3}_{\kappa}$, we have the explicit expression for the principal matrix thanks to the above explicit expression of the heat kernel (\ref{lowerboundheatkernelhperbolicspaces}),
\beqs
\Phi_{ij}(z)={1 \over 4 \pi} \left(\sqrt{\kappa^2 -z}-\sqrt{\kappa^2+\mu_{i}^{2} }\right) \delta_{ij} - (1-\delta_{ij}) \left( {\kappa \exp \left(- d(a_i,a_j)  \sqrt{\kappa^2 -z}\right) \over 4 \pi \sinh \left(\kappa d(a_i,a_j) \right) } \right)  \;, \label{PhiH3}
\eeqs
which clearly has the branch cut along $[\kappa^2, \infty)$. 
As for the two dimensional hyperbolic spaces $\mathbb{H}^{2}_{\kappa}$, we can also find the explicit expression of the principal matrix by interchanging the order of $t$ and $s$-integrations (Fubini's theorem),   
\beqs \begin{split}
\Phi_{ij}(z)  = & {1 \over 2 \pi} \left[ \psi \left( {1 \over 2} + \sqrt{-{z \over \kappa^2} + {1 \over 4}} \right)  - \psi \left( {1 \over 2} + \sqrt{{\mu_{i}^2 \over \kappa^2} +{1 \over 4}} \right) \right] \;\delta_{ij} \\[2ex] &  \hspace{4cm} - \, {1 \over 2 \pi} (1-\delta_{ij}) \; Q_{ {1 \over 2} + \sqrt{-{z \over \kappa^2} + {1 \over 4}} } \left(\cosh(\kappa d(a_i,a_j) )\right) \;,
\label{PhiH2} 
\end{split}
\eeqs
where we have used the integral representation of the digamma function \cite{Lebedev}
\beqs
\psi(z)=\int_{0}^{\infty} \left( {e^{-t} \over t} -{e^{- t z} \over 1-e^{-t}} \right) \; d t \;, \label{intreppsi}
\eeqs
for $\Re(z)>0$, and the integral representation of the Legendre function of second type \cite{Lebedev}
\beqs
Q_{\lambda}(\cosh a)=\int_{a}^{\infty} {e^{-(\lambda + {1 \over 2}) r} \over \sqrt{2 \cosh r -2 \cosh a}} \; d r \;, \label{intrepQ}
\eeqs
for real and positive $a$ and $\Re(\lambda)>-1$. From the above result (\ref{PhiH2}), we see that the branch cut is along $[\kappa^2/4, \infty)$, which completes the proof.
\end{proof}

For real values of $z$, the principal matrix $\Phi$ is symmetric, i.e., $\Phi(z)^{*}=\Phi(z)$, so that $\Phi(z)$ is a symmetric matrix-valued holomorphic function so that its eigenvalues and eigenprojections are holomorphic on the real axis due to the theorem 6.1 in \cite{Kato}. Throughout the paper, $*$ denotes the adjoint. One can also make the analytical continuation of the principal matrix $\Phi$ from the region $\mathcal{R}$ onto the largest possible set of complex plane except possibly the real axis.
As a consequence of this theorem, the operation of taking the derivative of the matrix elements of the principal matrix under the integral sign is justified. This operation without testing its validity was used in our previous works to find the flow of eigenvalues $\omega_n$ of the principal matrix (i.e., ${d \omega_n \over d E} >0$).

Now, we are going to give a new result about the essential spectrum of the Hamiltonian $H$ for compact manifolds in Proposition \ref{essential spectrum prop} and discuss some other spectral properties (\textit{partly} given in our previous work \cite{point interactions on manifolds2}) of our problem in Theorem \ref{compact spectrum}.

\begin{prop} \label{essential spectrum prop}
Let $(M,g)$ be a compact connected Riemannian manifold without boundary. Then, the essential spectrum of the operator $H$ is empty.
\end{prop}

The point interactions on flat spaces are known as finite rank perturbations to the free Hamiltonian so that the essential spectrum of the Hamiltonian is the same as the one of the free Hamiltonian. Hence, it is expected that this is also true on Riemannian manifolds. A technical proof of it is given in Appendix \ref{proof of prop1}.

\begin{thm} \label{compact spectrum} Let $(M,g)$ be a compact connected Riemannian manifold without boundary. Then, the spectrum of $H$ is purely discrete and contained in $(-\infty,0) \cup \{\sigma_l\} $, where $\{\sigma_l\} = \{ 0=\sigma_0 \leq \sigma_1 \leq \ldots \} $ and $\sigma_l \rightarrow \infty$ as $l \rightarrow \infty$, and each $\sigma_l$ has finite multiplicity. It has at most $N$ (negative) eigenvalues counting multiplicity and $-\nu^2$ ($\nu$ is real and positive) belongs to the negative part of the point spectrum iff $\det \Phi(-\nu^2)=0$ (characteristic equation). The multiplicity of the eigenvalue $-\nu^2$ equals to the multiplicity of the eigenvalue zero of the matrix $\Phi(-\nu^2)$. Moreover, let $E=-\nu_{k}^{2}$ be an eigenvalue of $H$, then the corresponding eigenfunctions $\psi_k(x)$ are given by
\beqs
\begin{split}
\psi_k(x) = & \left[ \sum_{i,j=1}^N \overline{A_i(\nu_k)}
\int_{0}^{\infty}  t
\;  K_t(a_i,a_j) e^{-t \nu_{k}^2} \,
A_j(\nu_k) \; dt  \right]^{-\frac{1}{2}} \; \\[2ex] & \hspace{5cm}\times
\int_{0}^{\infty} e^{-t
\nu_{k}^2} \sum_{i=1}^N A_i(\nu_k) K_{t}(a_i, x) \; d t \;,
\label{wavefunction heat delta} \end{split} \eeqs
where $(A_1,A_2,\ldots, A_N)$ are the eigenvectors with eigenvalue zero of the matrix $\Phi(-\nu_{k}^2)$ and  the ground state is nondegenerate and the corresponding eigenfunction can be chosen strictly positive.  
\end{thm}

The proof of this theorem was essentially given in our previous work \cite{point interactions on manifolds2} so we give it in Appendix \ref{proof of spectrumtheorem} for the completeness of the paper.

Since the Laplacian $-\triangle_g$ is symmetric and positive, its spectrum is contained in
$[0,\infty)$. When $M=\mathbb{R}^{D}$, then the spectrum of $-\triangle$ has no point spectrum. For a general noncompact Riemannian manifold $M$, the spectrum may include positive eigenvalues \cite{Donnelly0}. Nevertheless, under some mild conditions, it is expected  that there does not exist any positive eigenvalue with the finite multiplicity of $-\triangle_g$. This is stated as a conjecture in \cite{CL}: 
\begin{conj} \label{essential spectrum of laplacian}
Let $M$ be a complete noncompact Riemannian manifold with Ricci tensor bounded below ($\Ric(.,.) \geq  c \; g(.,.)$). Then, the essential spectrum of $-\triangle_g$ on functions is a connected subset of the positive real line $[a,\infty)$. 
\end{conj}

Following the same argument given in the proof of Proposition \ref{essential spectrum prop} and using the upper bound of the heat kernel for Cartan-Hadamard manifolds (\ref{heatkernel bound2}), it is easy to see that  the essential spectrum of $H$ is the same as that of $H_0$. In other words, the Hamiltonian $H$ is a compact perturbation to the free Hamiltonian $H_0$:
\begin{cor}
Let $M$ be Cartan-Hadamard manifold. Then, the point spectrum of $H$ in the positive real axis is empty{\color{red},} and the essential spectrum of the operator $H$ is 
$ \sigma_{ess}(H) =[a,\infty)$. In particular, $a=(D-1)^2 \kappa^2/4$ for $D$-dimensional hyperbolic manifolds of sectional curvature $-\kappa^2$ \cite{Donnelly 2}.
\end{cor}

\begin{prop} Let $M$ be compact manifold with Ricci tensor bounded below or Cartan-Hadamard manifold, and let $N(-\nu^2, \mu_1, \ldots, \mu_N)$ denote the number of bound states (counting multiplicities) of $H$ less than or equal to $-\nu^2$. Then, 

\beqs N(-\nu^2,\bar{\mu}, \ldots, \bar{\mu}) \leq N(-\nu^2,\mu_1, \ldots, \mu_N) \leq N(-\nu^2,\underline{\mu}, \ldots, \underline{\mu}) \;, \eeqs
where $ \bar{\mu}=\max_{1 \leq j \leq N} (\mu_j)$ and $\underline{\mu} = \min_{1 \leq j \leq N} (\mu_j)$. 
\end{prop}
\begin{proof} As a consequence of the Feynman-Hellmann theorem \cite{Feynman Hellmann1, Feynman Hellmann2} and the positivity of the heat kernel, it is easy to see that the derivative of the eigenvalues of the matrix $\Phi$ with respect to $\mu_k$ is $|A_{k}|^2\int_{0}^{\infty} K_t(a_k,a_k) (-2 t \mu_k)\; e^{-t \mu_{k}^2} \; d t $,
where  we have taken the derivative under the integral sign thanks to the Lemma \ref{holomorphic}. This is always negative and the proof is immediate from this fact.
\end{proof}

We now discuss the conditions on the number of bound states by starting from the special cases, where we have two point $\delta$-interactions separated by the distance $d$ on $\mathbb{R}^2$ and $\mathbb{R}^3$. This will be done by working out the characteristic equation $\det \Phi=0$. This problem is realized as a very elementary model for
ionized diatomic molecule $H_2^{+}$ and its one-dimensional version is discussed even in
the textbooks on quantum mechanics \cite{cohenbook}.


\section{On the Number of Bound States in Flat Spaces}
\label{On the Number of Bound States in Flat Spaces}

\begin{prop}  \label{numberof bound states flatspace} For $N=2$, if 
\begin{itemize}
\item[(i)] $\sqrt{\mu_1 \mu_2} \; d>2$ in $\mathbb{R}^2$ and 
\item[(ii)] $\sqrt{\mu_1 \mu_2} \; d>1 $ in $\mathbb{R}^3$.
\end{itemize} 
then there exist exactly two bound states. Otherwise, there exists precisely one bound state. \label{ConditionBoundstatesFlatresults}
\end{prop}

\begin{proof}
Since the result is well-known in the literature, we give a proof of it in Appendix \ref{proof theorem numberofbs}  in order to be self-contained.
\end{proof}

Above results seem to be a little different from the the ones given in \cite{Ogurisu 3} ($ \log d > 2\pi \alpha$ in two dimensions, $\alpha d > 1/4\pi$ in three dimensions, where $\alpha$ is related to the parameter $\mu$ implicitly). However, this is due to the different choice of the coupling constant (\ref{barecouplingdeltamanifolds}) and the computations there are performed in momentum space. Nevertheless, these results are essentially same. In one dimension, we do not need renormalization, so we have $\Phi_{ij}(-\nu^2)= ({1 \over \lambda_i} - {1 \over 2 \nu}) \delta_{ij} -(1-\delta_{ij}) {1 \over 2 \nu} \; e^{- d_{ij} \nu}$. The above analysis can be easily applied to this case and the condition for two bound states is given by $d>{ 1 \over \lambda_1} + {1 \over \lambda_2}$, which is exactly the same result as in \cite{Albeverio Nizhnik 2}.

It is not easy
to determine what condition must be met for the Hamiltonian with an arbitrary number of delta
centers to have $N$ bound states directly from the characteristic equation. In this case, the characteristic equation becomes much more complicated to work with. Moreover, there is no explicit expression for the the principal matrix $\Phi(-\nu^2)$ because there is no explicit expression for the heat kernel of a general Riemannian manifold.
In order to solve this problem in more generic class of manifolds, we essentially follow the idea established for the same problem in one dimension \cite{Ogurisu 1} and develop it onto the particular class of  Riemannian manifolds.


\section{Main Results} 
\label{Main Results}

\begin{cor} By Lemma \ref{holomorphic}, the principal matrix is real symmetric and continuously differentiable matrix-valued function on the complex plane with $\Re(z)<0$. 
\end{cor}

\begin{prop} (Theorem II.6.8, \cite{Kato})
Let $T(k) = (t_{ij}(k))_{i,j=1}^{N}$ be a real symmetric and continuously differentiable matrix. Suppose that $\lim_{k \rightarrow \infty} T(k)=diag(a_1,a_2, \ldots,a_N)$. Then, the following holds:
\begin{itemize}
\item[(i)] There exist $N$ continuously differentiable functions $\tau_i(k)$ that represent the repeated eigenvalues of the matrix $T(k)$.
\item[(ii)] $\lim_{k \rightarrow \infty} \tau_i(k)=a_i$ for all $i=1,\ldots,N$. 
\end{itemize}  \label{Kato}
\end{prop}

\begin{lem}
Let 
\beqs T_{ij}(-\nu^2) ={1 \over g(-\nu^2)} \; \Phi_{ij}(-\nu^2) =\begin{cases}
\begin{split}
{1\over {1 \over 2 \pi} \; \log \nu/\mu_i } \; \Phi_{ij} (-\nu^2)  \;,
\end{split}
& \mathrm{for} \; D=2 \\[2ex]
\begin{split}
{1 \over {1 \over 4 \pi} \left(\nu- \mu_i\right) } \; \Phi_{ij}(-\nu^2) \;,
\end{split}
& \mathrm{for} \; D=3\; ,
\end{cases} \label{asymPrincipal}\eeqs
for $\nu> \mu_i$. Then, there exist $N$ continuously differentiable functions ${\omega_i(-\nu^2) / g(-\nu^2)}$ that represent the eigenvalues of $T_{ij}(-\nu^2)$, where $\omega_i(-\nu^2)$ is the eigenvalue of the matrix $\Phi_{ij}(-\nu^2)$. Moreover, $\displaystyle{\lim_{\nu \rightarrow \infty}} {\omega_i(-\nu^2) \over g(-\nu^2)}=1$ for all $i$. \label{corollary asym}
\end{lem}

\begin{proof} Since the principal matrix $\Phi$ is symmetric, continuously differentiable matrix, so is $T$ for $\nu > \mu_i$. The off-diagonal elements of the principal matrix $\Phi(-\nu^2)$ vanishes as $\nu \rightarrow \infty$. This can be easily seen by Lebesgue dominated convergence theorem so that the order of limit and integral can be interchanged. This is possible since the term $K_{t}(a_i,a_j) \; e^{-t \nu^2} $ is dominated by the upper bounds of the heat kernel (\ref{heatkernel bound1}) and (\ref{heatkernel bound2}) multiplied by $e^{-\mu_i t}$ for all $t$ and $\nu > \mu_i \neq 0$. Therefore, we obtain $
\lim_{\nu \rightarrow \infty} \Phi_{ij}(-\nu^2) =0$ for $i \neq j$. Hence, $\lim_{\nu \rightarrow \infty} T_{ij}(-\nu^2) =0$ for $i \neq j$.

Using the lower bound of the diagonal heat kernel (\ref{heatkernellowerbound}), we can find the lower bound of the diagonal part of the principal matrix. It is easy to find the lower bound of the principal matrix for three dimensions due to the explicit expression of the heat kernel (\ref{lowerboundheatkernelhperbolicspaces}) for $D=3$. However, we need to estimate the closed expression of the heat kernel for the two dimensional hyperbolic manifolds given in (\ref{lowerboundheatkernelhperbolicspaces}). The diagonal lower bound for two dimensional hyperbolic manifolds of sectional curvature $-\kappa^2$ is given by \cite{Davies Mandouvalos}
\beqs
K_t(x,x) \geq  {1 \over 8 (4 \pi)^{3/2}} {e^{- \kappa^2 t/4} \over t \sqrt{1+ \kappa^2 t}} \;,
\label{mando}
\eeqs 
for all $t>0$ and $x \in M$. The constant factor in the above upper bound is not crucial for our purpose here (which was also absent in \cite{Davies Mandouvalos}). Using ${1 \over t \sqrt{1+\kappa^2 t}} \geq {\kappa^2 \over  (1+\kappa^2 t)^{3/2}}$ for all $t>0$ together with the integral representation of the complementary error function $\erfc$ (entry 3.369 in \cite{Gradshteyn})
\beqs
\int_{0}^{\infty} {e^{-a t} \over (b+t)^{3/2}} \; dt = {2 \over \sqrt{b}} -2\sqrt{\pi a} e^{a b} \; \erfc(\sqrt{a b}) \;,
\eeqs
for all $a,b>0$, we obtain
\beqs \label{lowerbounddiagonalPhi} \Phi_{ii} 
\geq
\begin{cases}
\begin{split}
{1 \over 32 \pi} \;   (\phi(\nu)-\phi(\mu_i))  \;,
\end{split}
& \mathrm{for} \; D=2 \\[2ex]
\begin{split}
 {1 \over 4 \pi} \left( \sqrt{\nu^2 + \kappa^2} - \sqrt{\mu_{i}^2 + \kappa^2}\right)  \;,
\end{split}
& \mathrm{for} \; D=3 \;,
\end{cases}
\eeqs
where $\phi(x)= \sqrt{ {x^2 \over \kappa^2} + {1\over 4}} \; e^{{x^2 \over \kappa^2} + {1\over 4}} \; \erfc( \sqrt{ {x^2 \over \kappa^2} + {1\over 4}})$. This shows that $\Phi_{ii} \rightarrow \infty$ as $\nu \rightarrow \infty$. We can find the asymptotic behaviour of the diagonal principal matrix $\Phi_{ii}$ as $\nu \rightarrow \infty$ as follows. Since the major contributions to the integral come from the neighbourhoods around $t=0$, it is natural to use the short time asymptotic expansion of the diagonal heat kernel (\ref{asymheat}) to get 
\beqs
\Phi_{ii} (-\nu^2) \sim 
\begin{cases}
\begin{split}
{1 \over 2 \pi} \; \log \nu/\mu_i \;,
\end{split}
& \mathrm{for} \; D=2 \\[2ex]
\begin{split}
 {1 \over 4 \pi} \left(\nu- \mu_i\right)  \;,
\end{split}
& \mathrm{for} \; D=3 \;,
\end{cases}
\eeqs
as $\nu \rightarrow \infty$. This motivates us to define a modified matrix (\ref{asymPrincipal}). 
Then, $\lim_{\nu \rightarrow \infty} T_{ij}(-\nu^2) = diag(1, \ldots, 1)$ so that it satisfies the hypothesis of the Theorem \ref{Kato}, so that the eigenvalues of the principal matrix $\Phi$ tends asymptotically to a positive function $g$ for large values of $\nu$ and this completes the proof. 
\end{proof}

\begin{lem}
If $ \Phi(-\nu_{*}^2)$ is negative definite with some $\nu_*> 0$, then we have $N$ bound states. \label{lemma negative definite}
\end{lem}
\begin{proof}
Due to the Lemma \ref{corollary asym}, $\omega_i(-\nu^2) >0$ for large enough $\nu$. According to the assumption of the lemma, $\omega_i(-\nu_{*}^2) <0$ for all $i$, then there exist at least $N$ number of $\nu_i$ such that $\omega_i(-\nu_{i}^2)=0$ for all $i$ due to the intermediate value theorem. Hence, it implies that $\det \Phi(-\nu_{i}^2)=0$, so that $-\nu{_i}^2$ is an eigenvalue. 
The monotonic behaviour of $\omega_i$'s guarantees that there exists exactly $N$ number of $\nu_i$ such that  $\omega_i(-\nu_{i}^2)=0$ for all $i$.
\end{proof}

Let $d=\displaystyle \min_{1 \leq i,j \leq N} \{d(x_i,x_j) ; i \neq j\}$ and $\mu= \displaystyle \min_{1 \leq i \leq N} \mu_i$. Using Lemma \ref{lemma negative definite} and the
following Gerschgorin's theorem, we can prove our main theorem. In particular, if we restrict the class of Cartan-Hadamard manifold to the hyperbolic spaces of sectional curvature $-\kappa^2$, we obtain an explicit criterion for the existence of $N$ bound states
in terms of  $d$, $\mu$ and $\kappa$:

\begin{prop} (Theorem 6.1.1, \cite{Horn})
All eigenvalues of a matrix $T$ are contained in the union of Gerschgorin's disks
\beqs G_i = \left\{ z \in \mathbb{C}; |z-T_{ii}| \leq \sum_{j \neq i} |T_{ij}| \right\} \eeqs
for $i=1, \ldots , N$.
\end{prop}

The following theorem is one of the main results in this paper: 

\begin{thm} \label{Gersnumberbound}
\begin{itemize}

\item[(i)] If there exists $\nu_* >0$ such that 
\beqs \Phi_{ii}(-\nu_{*}^2) + \sum_{j \neq i} |\Phi_{ij}(-\nu_{*}^2)| <0 \;, \label{generalNboundstatecondition}\eeqs
then there are $N$ bound states.

In particular, 

\item[(ii)] If 
\beqs
\exp \left( d \sqrt{\kappa^2 +\mu^2} -1 \right) \; \left( { \sinh \kappa d \over \kappa d } \right) > (N-1) \;,
\label{condition for N bound states 3D}
\eeqs
holds in $\mathbb{H}^3$, then there are $N$ bound states. 

\item[(iii)] If either 
\beqs
{1 \over 2} + \sqrt{{\mu^2 \over \kappa^2} +{1 \over 4}} \geq e  \hspace{0.5cm}  & \mathrm{and} & \hspace{0.5cm} (N-1) < {\kappa d \over 2 e} \;\left( {1 \over 2} + \sqrt{{\mu^2 \over \kappa^2} +{1 \over 4}} \right)\label{2dboundcond1} 
\eeqs
or 
\beqs
 {1 \over 2} + \sqrt{{\mu^2 \over \kappa^2} +{1 \over 4}} < e  \hspace{0.5cm}  & \mathrm{and} & \hspace{0.5cm} (N-1) < {\kappa d \over 2} \log \left( {1 \over 2} + \sqrt{{\mu^2 \over \kappa^2} +{1 \over 4}}\right) \label{2dboundcond2} \;.
\eeqs
holds in $\mathbb{H}^2$, then there are $N$ bound states.
\end{itemize}

\end{thm}

\begin{proof} Let
\beqs G_i(-\nu^2) = \left[ \Phi_{ii}(-\nu^2) - \sum_{j \neq i} |\Phi_{ij}(-\nu^2)|, \;  \Phi_{ii}(-\nu^2) + \sum_{j \neq i} |\Phi_{ij}(-\nu^2)| \right] \;. \eeqs
Then, Gerschgorin's theorem implies that $\omega_i (-\nu^2) \in \cup_{j=1}^{N} G_j(-\nu^2)$ for all $i$. Thus, all the eigenvalues $\omega_i(-\nu_{*}^2)$ are negative and the hypothesis of Lemma \ref{lemma negative definite} holds, which then proves the statement $(i)$. 

The statement (ii) can be proved as a corollary of (i). Let us first notice that $\displaystyle \max_{1 \leq i \leq N} G_i(-\nu^2) \leq \max_{1 \leq i \leq N} \Phi_{ii}(-\nu^2) + (N-1) \max_{1 \leq i \leq N} \max_{1 \leq j \neq i \leq N} |\Phi_{ij}(-\nu^2)|$. For this to be negative, it is necessary that the first term $\max_{1 \leq i \leq N} \Phi_{ii}(-\nu^2)$ must be negative. 
For the three dimensional hyperbolic spaces $\mathbb{H}^{3}_{\kappa}$, 
it is easy to see that imposing the following condition
\beqs
\left(\sqrt{\kappa^2+\nu^2}-\sqrt{\kappa^2+\mu^{2} }\right) +(N-1) \left( { \kappa \exp \left(- d \sqrt{\kappa^2+\nu^2}\right) \over \sinh \kappa d } \right) < 0   \label{H3condition}
\eeqs
implies the condition (\ref{generalNboundstatecondition}) for the principal matrix (\ref{PhiH3}) at $z=-\nu^2$. In order to find the sufficient criterion for this to be negative, we must necessarily have $\nu<\mu$. For this reason, let us define the function $F_1(\nu)$ to be the right-hand side of (\ref{H3condition}). Let $\nu_c$ satisfies $F_{1}'(\nu_c)=0$. Since $F_{1}''(\nu) > 0$, we obtain that 

\beqs
\inf_{0 < \nu <\mu} F_1(\nu) =
\begin{cases}
\begin{split}
F_1(\mu) \;,
\end{split}
& \mathrm{if} \; \nu_c \geq \mu \\[2ex]
\begin{split}
F_1(\nu_c) \;,
\end{split}
& \mathrm{if} \; 0 \leq \nu_c < \mu\;, \\[2ex]
\begin{split}
F_1(0) \;,
\end{split}
& \mathrm{otherwise} \;
\end{cases}
\eeqs
Using this, we can see that $\inf_{0 < \nu <\mu} F_1(\nu)<0$ if and only if
\beqs
(N-1) < \exp\left( d \sqrt{\kappa^2 + \mu^2} -1 \right) \; {\sinh \kappa d \over \kappa d} \;.
\eeqs
Thus, we arrive at the claimed inequalities given in the statement (ii).

For two dimensional hyperbolic spaces $\mathbb{H}^{2}_{\kappa}$, if we impose the condition
\beqs
\left[ \psi \left( {1 \over 2} + \sqrt{{\nu^2 \over \kappa^2} + {1 \over 4}} \right) - \psi \left( {1 \over 2} + \sqrt{{\mu^2 \over \kappa^2} +{1 \over 4}} \right) \right]  + (N-1) \; Q_{ {1 \over 2} + \sqrt{{\nu^2 \over \kappa^2} + {1 \over 4}} } \left(\cosh \kappa d \right) < 0 \;, 
\label{H2condition}
\eeqs
it implies (\ref{generalNboundstatecondition}) for the principal matrix (\ref{PhiH2}) at $z=-\nu^2$ since digamma function $\psi(x)$ is an increasing function for all real positive $x$ whereas the Legendre function of second type $Q_\lambda(x)$ is decreasing function for all real $x>1$. Similar to the three dimensional case, we must necessarily have $\nu < \mu$ since $Q_\lambda$ is always positive. We first find an upper bound to the first term in the left hand side of the inequality (\ref{H2condition}) using the integral representation (\ref{intreppsi}) and the bound ${1 \over 1-e^{-t}} \geq {1 \over t}$ for all $t$ (Note that the difference of the exponentials are always negative due to $\nu<\mu$). Similarly, we can also find an upper bound for the Legendre function of second type in the second term of the inequality (\ref{H2condition}) by using its integral representation (\ref{intrepQ}) together with the bound $\cosh r -\cosh a \geq {1 \over 2} (r^2 -a^2)$. Hence, imposing 
\beqs
\log\left( {{1 \over 2} + \sqrt{{\nu^2 \over \kappa^2} + {1 \over 4}} \over  {{1 \over 2} + \sqrt{{\mu^2 \over \kappa^2} +{1 \over 4}}} }\right)  + {2 (N-1) \over \kappa d} {1 \over 1  + \sqrt{{\nu^2 \over \kappa^2} + {1 \over 4}}} < 0 \label{H2condition2}
\eeqs
implies (\ref{H2condition}). Here we have used the integral representation \cite{Lebedev} of modified Bessel function 
of the third kind $K_0 (x)= \int_{1}^{\infty} e^{- x t}/\sqrt{t^2 -1} \; d t$ and the upper bound for it $K_0(x) < {e^{-x/2} \over x/2} < {1 \over x/2}$ for all real $x>0$, which was proved in our earlier work \cite{point interactions on manifolds1}. It is hard to find the infimum of the left-hand side of (\ref{H2condition2}), so we use the simplified  and the following less sharper bound to (\ref{H2condition}):
\beqs
\log\left( {t(\nu) \over t(\mu)}\right)  +  {c (N-1)  \over t(\nu)} < 0 \;, \label{infimum2}
\eeqs
where we have defined $t=t(\nu)= {1 \over 2} + \sqrt{{\nu^2 \over \kappa^2} + {1 \over 4}} $  and $c=2/\kappa d$. 
Similar to the three dimensional case, we define the function $F_2(\nu)$ to be the left-hand side of the above inequality (\ref{infimum2})
\beqs
F_2(\nu)=f(t) = \log(t/t(\mu)) + {c (N-1)\over t} \;. \label{F_2} 
\eeqs
Then, it is easy to find the infimum of the function $F_2(\nu)$:
\beqs
\inf_{0 < \nu <\mu} F_2(\nu) = \inf_{t(0) < t < t(\mu)} f(t) =
\begin{cases}
\begin{split}
f(t(\mu)) \;,
\end{split}
& \mathrm{if} \;\; t_c \geq t(\mu) \\[2ex]
\begin{split}
f(t_c) \;,
\end{split}
& \mathrm{if} \;\; t(0)=1 \leq t_c < t(\mu)\;, \\[2ex]
\begin{split}
f(t(0)) \;,
\end{split}
& \mathrm{otherwise} \;,
\end{cases}
\eeqs
where $t_c=c(N-1)$ is the solution of $f'(t_c)=0$. This infimum of $F_2(\nu)$ is negative for some $\nu$ if and only if
either (\ref{2dboundcond1}) or (\ref{2dboundcond2}) holds, then there exists $N$ bound states.

\end{proof}

\begin{rem} Let us consider the limiting case, where the sectional curvature $-\kappa^2$ of the hyperbolic manifolds approaches zero. In order to compare the result with the flat space results, we consider for simplicity two point $\delta$-interactions with the same strength ($\mu_1=\mu_2=\mu$).  From the explicit bounds given in the statements $(ii)$ and $(iii)$ converge to the bounds $\mu d > 2 e$ 
in two dimensions and $\mu d > 1$ 
in three dimensions as $\kappa \rightarrow 0$. These results are  pretty consistent with the ones given in the Proposition \ref{ConditionBoundstatesFlatresults} except that the two-dimensional results are slightly different. However, this is due to the estimations that we made in order to be able to find an  analytical result.  Furthermore, for $N=1$ we always have one bound state no matter how small $\kappa$ and other parameters are. 
\end{rem}

Actually, similar criteria can also be found for compact manifolds with Ricci tensor bounded from below and Cartan-Hadamard manifolds by using the heat kernel upper and lower bounds (\ref{Cheeger-Yau remark}), (\ref{heatkernel bound1}) and (\ref{heatkernel bound2}). However, the results would depend on the unknown coefficients $C_1, C_2,C_3$ and $C_4$ and this would not be useful from the physical point of view.

We can also improve the above results by the following proposition, the result of which is not necessary to understand the next part of the paper. 

\begin{prop}
Suppose that $\mu_j \geq \mu_2 > \mu_1=\mu$ for all $i \geq 3$ (order it by renumbering $\mu_i$'s) and that there exists $\nu_*>0$ such that $\Phi_{11}(-\nu_{*}^2)<0$, $\Phi_{22}(-\nu_{*}^2)<0$,  and
\beqs 
\left(\Phi_{11}(-\nu_{*}^2) + \Phi_{22}(-\nu_{*}^2) \right) - \Bigg[ \left(\Phi_{11}(-\nu_{*}^2) - \Phi_{22}(-\nu_{*}^2) \right)^2 + 4 \left(\sum_{j \neq k} |\Phi_{jk}(-\nu_{*}^2)|\right)^2 \Bigg]^{1/2} < 0 \;,
\eeqs
then there exist $N$ bound states. 
\end{prop}

\begin{proof} We will use Brauer-Cassini's theorem (Theorem 6.4.7 in \cite{Horn}): All eigenvalues of $\Phi$ are located in the union of  $N(N-1)/2$ ovals of Cassini $K_{j,k}$
\beqs \bigcup_{j \neq k}^{N} \left\{z \in \mathbb{C} |z-\Phi_{jj}| \; |z-\Phi_{kk}| \leq \sum_{i \neq j} |\Phi_{ij}|\sum_{i \neq k} |\Phi_{ik}| \right\} \;.
\eeqs
Then, it holds that 
\beqs \begin{split} &  \bigcup_{j \neq k}^{N}  K_{j,k} \cap \mathbb{R} \leq {1 \over 2} \left(\Phi_{11}(-\nu^2) + \Phi_{22}(-\nu^2) \right) \\[2ex]  & \hspace{2cm} \nonumber \hspace{2cm}- \Bigg[ \left(\Phi_{11}(-\nu^2) - \Phi_{22}(-\nu^2) \right)^2 +  4 \left(\sum_{j \neq k} |\Phi_{jk}(-\nu^2)|\right)^2  \Bigg]^{1/2} \;. \end{split} \eeqs
If the right hand side is negative for some $\nu_*>0$, then the assumptions of  Lemma \ref{lemma negative definite} holds, so we obtain the desired result.  
\end{proof}
This result is stronger than part $(i)$ in Theorem \ref{Gersnumberbound} since
\beqs \begin{split} &
\Phi_{ii}(-\nu^2) + \sum_{j \neq i} |\Phi_{ij}(-\nu^2)| \leq \Phi_{11}(-\nu_{*}^2) + \sum_{j \neq i} |\Phi_{ij}(-\nu_{*}^2)| \\[2ex]  & \hspace{1cm} \leq    {1 \over 2} \left(\Phi_{11}(-\nu^2) + \Phi_{22}(-\nu^2) \right) - \Bigg[ \left(\Phi_{11}(-\nu^2) - \Phi_{22}(-\nu^2) \right)^2 +  4 \left(\sum_{j \neq k} |\Phi_{jk}(-\nu^2)|\right)^2  \Bigg]^{1/2} \;.
\end{split}
\eeqs


\section{A Relativistic Extension of the Model on $\mathbb{R}^2$ and $\mathbb{H}^2$}
\label{ARelativisticExtensionoftheModelon}

One relativistic extension of the above model on two dimensional Riemannian manifolds was first considered in \cite{Caglar1}. Here we are first going to summarize the basic idea of the construction of the model.  In this model,
relativistic Klein-Gordon particles interacts with finitely many localized sources on $M$. We use the units such that $\hbar=c=1$. The second quantized regularized Hamiltonian is formally given by
\beqs \begin{split}
&  H_{\epsilon}= {1 \over 2} \int_{M} d_{g}^2 x \; : \left[ \Pi^2 + \phi(x) \left(-\Delta_g +m^2 \right) \phi(x) \right]:  \\[2ex] & \hspace{3cm}- \sum_{i=1}^{N} g_i(\epsilon) \int_{M}   d_{g}^2 x \; K_{\epsilon/2}(a_i, x) \phi^{(-)}(x) \; \int_{M}   d_{g}^2 y \; K_{\epsilon/2}(a_i,y) \phi^{(+)}(y) \;, \label{relativisticH_epsilon}
\end{split}
\eeqs
where $\Pi(x)=\partial_0 \phi(x,t)$ at $t=0$ (also $\phi(x,0)=\phi(x)$) and the normal ordering is denoted by colon, and $\phi^{(-)}(x)$ is the positive frequency part of the real bosonic field operator, given in terms of the creation operator $a^{*}_{\sigma}$ (the index $\sigma$ is the analog of the momentum label in flat spaces):
\beqs
\phi^{(-)}(x) & = & \sum_{\sigma} { a^{*}_{\sigma} \; \overline{f_{\sigma}(x)}\over \sqrt{\omega(\sigma)}} \cr \omega^{2}_{\sigma} &=& \lambda(\sigma) + m^2
\eeqs
and $f_{\sigma}(x)$ are the orthonormal complete set of eigenfunctions of Laplace-Beltrami operator in $L^2(M)$, i.e., $-\Delta_g f_{\sigma}(x) = \lambda(\sigma) f_{\sigma}(x)$. This is a relativistic many-body problem, where the number of particles are conserved. Here is the idea of the paper \cite{Caglar1}: First, fictitious operators $\chi_i$ obeying ortho-fermion algebra are introduced 
\beqs
\chi_i \; \chi^{*}_j & = & \delta_{ij} \Pi_0 \cr \chi_i \; \chi_j & = & 0 \cr \sum_{i=1}^{N} \chi^{*}_i \; \chi_i & = & \sum_{i=1}^{N} \Pi_i = \Pi_1
\eeqs
where $\Pi_0$, $\Pi_1$ are the projection operators onto no ortho-fermion and 1-ortho-fermion states, respectively. Then, the following augmented operator is defined in matrix form on the augmented symmetrized Fock space $\mathcal{F}_s(H) \oplus \mathcal{F}_s(H) \otimes \mathbb{C}^N$:
\beqs
\begin{bmatrix}
     (H_0-z) \Pi_0 &  \sum_i \int_M d_{g}^2 x K_{\epsilon/2}(a_i,x) \phi^{(-)}(x) \chi_i \\[2ex]
     \sum_j \int_M d_{g}^2 y K_{\epsilon/2}(a_j,y) \phi^{(+)}(y) \chi_{j}^{*} &  \sum_{k,l} {\delta_{kl} \over g_k} \chi^{*}_k \chi_l
\end{bmatrix} 
\eeqs
Then, there are two apparently different but equivalent formula for the projection of the inverse of the above operator onto no ortho-fermion subspace and gives an explicit formula for regularized resolvent of our original Hamiltonian:
\beqs \begin{split}
 & (H_\epsilon -z)^{-1} = (H_0-z)^{-1} + (H_0-z)^{-1} \sum_{i=1}^{N} \int_M d_{g}^2 x \;  K_{\epsilon/2}(a_i,x) \phi^{(-)}(x) \; \Phi^{-1}_{\epsilon}(E) \; \\[2ex]  & \hspace{7cm} \times \int_M d_{g}^2 y \; K_{\epsilon/2}(a_i,y) \phi^{(+)}(y) \; (H_0-z)^{-1}
\end{split}
\eeqs
where the regularized principal operator is defined as
\beqs \begin{split}
& \Phi_{\epsilon}= \sum_{i=1}^{N} {1 \over g_i(\epsilon)} \chi^{*}_i \chi_i - \sum_{i,j=1}^{N} 
\int_M d_{g}^2 y \; K_{\epsilon/2}(a_i,y) \phi^{(+)}(y) \; (H_0-z)^{-1} \\[2ex] & \hspace{7cm} \times \int_M d_{g}^2 x \;  K_{\epsilon/2}(a_i,x) \phi^{(-)}(x) \; \chi^{*}_i \chi_j \;.
\end{split}
\eeqs
After normal ordering the above regularized principal operator
and considering the single boson particle states, and choosing the coupling constants $g_i(\epsilon)$
\beqs
{1 \over g_i(\epsilon)} = {1 \over \sqrt{\pi}} \int_{0}^{\infty} d s \;  e^{-s^2/4} \; \int_{\epsilon}^{\infty} d u \; e^{s \mu_i \sqrt{u}} \; e^{-u m^2} K_u(a_i,a_i)
\eeqs
we obtain a nontrivial limit of the resolvent formula restricted to single boson state:
\beqs
(H-z)^{-1}= (H_0-z)^{-1} +  \sum_{i,j=1}^{N} (H_0-z)^{-1} \phi^{(-)}(a_i) \; \Phi^{-1}_{ij}(z) \; \phi^{(+)}(a_i) (H_0-z)^{-1}
\eeqs
where
\beqs \label{relativisticPhi}
\Phi_{ij}(z) = 
\begin{cases}
\begin{split}
{1 \over \sqrt{\pi}} \int_{0}^{\infty} d s \;  e^{-s^2/4} \; \int_{0}^{\infty} d u \; \left(e^{s \mu_i \sqrt{u}} - e^{s z \sqrt{u}}\right) \; e^{-u m^2} K_u(a_i,a_i)
\end{split}
& \mathrm{if} \;\; i = j \\[2ex]
\begin{split}
- {1 \over \sqrt{\pi}} \int_{0}^{\infty} d s \;  e^{-s^2/4} \; \int_{0}^{\infty} d u \; e^{s z \sqrt{u}} \; e^{-u m^2} K_u(a_i,a_j)
\end{split}
& \mathrm{if} \;\;  i \neq j.
\end{cases} \;,
\eeqs
and $\mu_i<m$ is the experimentally measured bound state energy of the single relativistic boson in the single $i$th delta center. In order to prevent pair productions, we must have $-m < \Re{(E)} < m$. The upper bound must be due to the attractiveness of the potential (see \cite{Caglar1} for the details). Similar to the non-relativistic version of this problem, we have showed that there exists a self-adjoint operator associated with the above resolvent formula in \cite{caglar}. Moreover, the eigenvalues $\omega$ of the principal matrix flow according to ${d \omega \over d E} <0$.

\begin{lem} The principal matrix given in (\ref{relativisticPhi})  for compact manifolds with Ricci tensor bounded from below and for Cartan-Hadamard manifolds is a matrix-valued holomorphic function on the complex plane, where $\Re{(E)}<m$. 
\end{lem}

\begin{proof}
In order to show analyticity of the above principal matrix, we first make a change of variable $s=t/\sqrt{u}$, then we have
\beqs \label{relativisticPhi}
\Phi_{ij}(E) = 
\begin{cases}
\begin{split}
{1 \over \sqrt{\pi}} \int_{0}^{\infty} d t \;  \left(e^{\mu_i t} - e^{t E} \right) \;  \int_{0}^{\infty} d u \;  \; { e^{-t^2/4u} \; e^{-u m^2} K_u(a_i,a_i) \over \sqrt{u}}
\end{split}
& \mathrm{if} \;\; i = j \\[2ex]
\begin{split}
- {1 \over \sqrt{\pi}} \int_{0}^{\infty} d t \;  e^{E t} \; \int_{0}^{\infty} d u \; { e^{-t^2/4u} \; e^{-u m^2} K_u(a_i,a_j) \over \sqrt{u}}
\end{split}
& \mathrm{if} \;\; i \neq j.
\end{cases} \;.
\eeqs
It is easy to see that $u$-integrations are uniformly convergent for all $t \in (0,\infty)$ using  Weierstrass's $M$ test. Then, the result of $u$ integrations are continuous functions of $t$. Hence, using the same line of arguments in the proof the Theorem \ref{holomorphic}, one can show that all the assumptions of the theorem given in the proof of Theorem \ref{holomorphic} are satisfied.
\end{proof}
\begin{lem}
Let
\beqs T_{ij}(E) ={1 \over g(E)} \; \Phi_{ij}(E) = {2\pi \over \log \left({ m- E \over m-\mu_i} \right)} \; \Phi_{ij}(E) \label{asymRelPrincipal}
\eeqs
for $E$ is real and $E<m$. Then, there exist $N$ continuously differentiable functions ${\omega_i(E) / g(E)}$ that represent the eigenvalues of $T_{ij}(E)$, where $\omega_i(E)$ is the eigenvalue of the matrix $\Phi_{ij}(E)$. Moreover, $\lim_{E \rightarrow -\infty} {\omega_i(E)/g(E)} =1$ for all $i$. 
\end{lem}
The proof is essentially the same as the non-relativistic version of the model.

\begin{thm}
\begin{itemize}

\item[(i)] If there exists a real $E = E_*$ and $E_*>\mu_i$, and $-m<E_*<m$ such that 
\beqs \Phi_{ii}(E_*) + \sum_{j \neq i} |\Phi_{ij}(E_*)| <0 \;, \label{relativisticgeneralNboundstatecondition}
\eeqs
then there are $N$ bound states.

\item[(ii)]  If 
\beqs
(N-1) < { d\, (m-\mu) \over e} \label{2drboundcond1} 
\eeqs
%
%
%
holds in $\mathbb{R}^2$, then there are $N$ bound states.

\item[(iii)] If either 
\beqs \hskip-1cm
\sqrt{m^2 + \kappa^2/4}- \mu \geq e \; (\sqrt{m^2 + \kappa^2/4}- m ) \; & \mathrm{and} & \; (N-1) < {d \over 24(4 \pi)^{3/2}} {  (\sqrt{m^2 + \kappa^2/4}-\mu) \over e} \label{2drboundcond1} 
\eeqs
or 
\beqs
\hskip-1cm
 (\sqrt{m^2 + \kappa^2/4}-\mu) & < & e (\sqrt{m^2 + \kappa^2/4}- m) \;\;\; \mathrm{and}  \cr \; (N-1) & < & {d \; (\sqrt{m^2 + \kappa^2/4}- m) \over 24(4 \pi)^{3/2}} \, \log ({\sqrt{m^2 + \kappa^2/4}-\mu \over (\sqrt{m^2 + \kappa^2/4}- m) }) \label{2drboundcond2} \;
\eeqs
holds in $\mathbb{H}^2$, then there are $N$ bound states.
\end{itemize}
\end{thm}

\begin{proof}
The proof of part (i) is the same as that of the part (i) of  Theorem \ref{Gersnumberbound}. The principal matrix in $\mathbb{R}^2$ is given by
\beqs
\Phi_{ij}(E)=\begin{cases}
\begin{split}
{1 \over 2\pi} \log \left( {m-E \over m- \mu_i} \right) 
\end{split}
& \mathrm{if} \;\; i = j \\[2ex]
\begin{split}
- {1 \over 2\pi}  \int_{0}^{\infty} d s \; { e^{-d_{ij} (m\sqrt{s^2 +1} -E s)} \over \sqrt{s^2 +1}}
\end{split}
& \mathrm{if} \;\; i \neq j.
\end{cases} \;. \label{relativisticprincipalflat}
\eeqs
We first find an upper bound on the off-diagonal term $|\Phi_{ij}|$
\beqs
|\Phi_{ij}(E)| \leq  {1 \over 2\pi}  \int_{0}^{\infty} d s \; e^{-d_{ij} (m -E) s} = {1 \over 2\pi d_{ij} (m-E)} \leq {1 \over 2\pi d (m-E)} \;. 
\eeqs
If we impose the condition
\beqs
{1 \over 2\pi} \log \left( {m-E \over m- \mu} \right) +  {(N-1) \over 2\pi d (m-E)} < 0 \;, \label{relativisticflatinequality}
\eeqs
it implies (\ref{relativisticgeneralNboundstatecondition}) for the principal matrix (\ref{relativisticprincipalflat}). Since the left-hand side of the above inequality (\ref{relativisticflatinequality}) is the same as (\ref{infimum2}) except for the form of the function $t(\nu)$ and $c$. Hence, the result is straightforward.

An upper bound of the heat kernel on $\mathbb{H}^2$ has been calculated in \cite{Davies Mandouvalos} without the constant coefficient since it was irrelevant for their purposes. If we follow the same line of arguments in the proof, it is not difficult to compute the constant sharply and obtain 
\beqs \begin{split}
K_t(x,y) & <  {4 \sqrt{2} e^{-3/8} (16 \sqrt{2} + 4 \sqrt{\pi}) \over (4\pi)^{3/2}} (1+\kappa d(x,y)) {\exp\left( - {d^2(x,y) \over 4t} -{\kappa d(x,y) \over 2} -{\kappa^2 t \over 4}\right) \over t \sqrt{1+\kappa d(x,y) + \kappa^2 t}}  \\[2ex] & <  {3 \over t} \; \exp\left( - {d^2(x,y) \over 4t} -{\kappa^2 t \over 4}\right) \;,
\end{split}
\eeqs
for all $x,y \in M$ and $t>0$. By changing the variable $s=t/\sqrt{u}$ in the off-diagonal part of the principal matrix and using the above bound for the heat kernel we get
\beqs
|\Phi_{ij}(E)| \leq 6 \int_{0}^{\infty} d s \; e^{-d(a_i,a_j) s (\sqrt{m^2 + \kappa^2/4} - E)} = {6 \over d(a_i,a_j) \left(\sqrt{m^2 + \kappa^2/4} - E \right)} \;.
\eeqs
Using the lower bound of the heat kernel (\ref{mando}) and the fact that $(1+\kappa^2 t)^{-1/2} \leq e^{-\kappa^2 t /2}$ for the diagonal part of the principal matrix, and imposing the condition
\beqs
{1 \over 4(4\pi)^{3/2}} \log \left( {E- \sqrt{m^2 + \kappa^2/4} \over \mu - \sqrt{m^2 + \kappa^2/4} }\right) + {6 (N-1) \over d \left(\sqrt{m^2 + \kappa^2/4} - E \right)} < 0 \;, \label{relativisticNboundstatecondition}
\eeqs
it implies  (\ref{relativisticgeneralNboundstatecondition}). The function on the left-hand side of the inequality is the same as in (\ref{infimum2}) except for the form of $t(\nu)$ and $c$, so the result is immediate.

\end{proof}
\section*{Acknowledments} 
I would like to thank Osman Teoman Turgut for useful discussions.  I also thank the anonymous referee for his/her valuable comments and suggestions which improve the paper.



\section*{Appendix A: A Proof of Theorem 1} \label{proofoffirsttheorem}

\setcounter{equation}{0}

\renewcommand{\theequation}{A-\arabic{equation}}

It is well-known that Laplacian $\triangle_g$ is essentially self-adjoint on complete Riemannian manifolds without boundary so there exists a unique self-adjoint extension of the Laplacian. Since the interaction term in $H_\epsilon$ is a bounded finite rank symmetric perturbation to the Laplacian,  $H_\epsilon$ is self-adjoint on the same domain of the Laplacian.  We first find the resolvent associated with the regularized Hamiltonian (\ref{regularized H}). Let $\tilde{K}_{\epsilon}(x,a_i)=\sqrt{\lambda_{i}(\epsilon)} \; K_{\epsilon}(x,a_i)$. This rescaling is necessary to preserve the symmetry property of the integral kernel. Then, a simple computation from the inhomogenous equation $(H_\epsilon -z) \psi = f$ for some $z$, $\Im(z)  \neq 0$ shows that
\beqs \psi(x) = \int_{M} R_0(x,x')\left(\sum_{i=1}^{N} (\tilde{K}_{\epsilon}^{i},\psi) \; \tilde{K}_{\epsilon}(x',a_i) +f(x')\right) \; d^{D}_{g} x' \label{psi} \;, \eeqs
where  $(\tilde{K}_{\epsilon}^{i},\psi)=\int_{M} \tilde{K}_{\epsilon}(x,a_i) \psi(x)\; d^{D}_{g} x$. If we multiply the above equation by $\tilde{K}_{\epsilon}(x,a_j)$ and integrate with respect to $x$, we get 
\beqs \sum_{j=1}^{N} A_{ij}^{\epsilon} (z) \; (\tilde{K}_{\epsilon}^{j}, \psi) = \iint_{M^2} R_0(x,x'|z) \tilde{K}_{\epsilon}(x,a_i) f(x') \; d_{g}^{D} x \; d_{g}^{D} x' \;, \label{regA1}\eeqs
for all $i=1, \ldots, N$. Here we have defined
\beqs \label{A} A_{ij}^{\epsilon} (z) =
\begin{cases}
\begin{split}
1- \iint_{M^2} R_0(x,x'|z) \tilde{K}_{\epsilon}(x,a_i) \tilde{K}_{\epsilon}(x',a_i)\; d_{g}^{D} x \; d_{g}^{D} x'
\end{split}
& \mathrm{if} \;\; i = j \\[2ex]
\begin{split}
- \iint_{M^2} R_0(x,x'|z) \tilde{K}_{\epsilon}(x,a_i) \tilde{K}_{\epsilon}(x',a_j)\; d_{g}^{D} x \; d_{g}^{D} x'
\end{split}
& \mathrm{if} \;\; i \neq j.
\end{cases} \;.
\eeqs
Solving $(\tilde{K}_{\epsilon}^{i},\psi)$ from (\ref{regA1}) and substituting this back into (\ref{psi}), we obtain
\beqs \begin{split} \psi(x) = \int_{M} R_0(x,x'|z) f(x') \; d_{g}^{D}x' + \sum_{i,j=1}^{N} & \iiint_{M^3} R_0(x,y|z) \tilde{K}_{\epsilon}(y,a_i) \left[A^{-1}(z)\right]_{ij} \\[2ex] & \times \, \tilde{K}_{\epsilon}(x',a_j) R_0(x',y') f(y') \; d_{g}^{D} x' \; d_{g}^{D} y \; d_{g}^{D} y' \end{split} \;. \eeqs
Inserting an $N \times N$ diagonal matrix $\mathbb{E} \mathbb{E}^{-1}$ before and after the matrix $A^{-1}$, where $\mathbb{E}_{ij}=\sqrt{\lambda_i(\epsilon)} \delta_{ij}$, yields
\beqs \begin{split} \psi(x) = \int_{M} R_0(x,y|z) f(y) \; d_{g}^{D}y + \sum_{i,j=1}^{N} & \iiint_{M^3} R_0(x,y'|z) K_{\epsilon}(y',a_i) \left[\Phi^{-1}(z)\right]_{ij} \\[2ex] & \times \, K_{\epsilon}(x',a_j) R_0(x',y) f(y) \; d_{g}^{D} x'  \; d_{g}^{D} y' \; d_{g}^{D} y \;, \end{split} \label{R_e} \eeqs
where
\beqs \label{Phi(epsilon)} \Phi_{ij}^{\epsilon} (z) =
\begin{cases}
\begin{split}
{1 \over \lambda_i(\epsilon)} - \iint_{M^2} R_0(x,x'|z) K_{\epsilon}(x,a_i) K_{\epsilon}(x',a_i)\; d_{g}^{D} x \; d_{g}^{D} x'
\end{split}
& \mathrm{if} \;\; i = j \\[2ex]
\begin{split}
- \iint_{M^2} R_0(x,x'|z) K_{\epsilon}(x,a_i) K_{\epsilon}(x',a_j)\; d_{g}^{D} x  \; d_{g}^{D} x'
\end{split}
& \mathrm{if} \;\; i \neq j \;.
\end{cases}
\eeqs
The resolvent (\ref{R_e}) has a nontrivial limit if and only if the diagonal terms of $\Phi$ converge to a nontrivial limit as $\epsilon \rightarrow 0^+$. For this reason, one can choose the coupling constants $\lambda_i(\epsilon)$ as in (\ref{barecouplingdeltamanifolds}) so that the limit of the resolvent (\ref{R_e}) converges to (\ref{resolvent}). Here the above matrix (\ref{Phi(epsilon)}) converges to the matrix (\ref{phiheat3}), called principal matrix.

However, it is not obvious at this stage that the operator $R$ obtained from the above limiting procedure is actually a resolvent of a densely defined closed operator. In Euclidean case, one can prove that the operator $R$ given in (\ref{resolvent}) is the resolvent of a closed operator by first going to Fourier space and then showing that the limit is injective \cite{Albeverio 1} since the pseudo-resolvent is a resolvent of a closed operator if and only if $\ker
(R) = \{0 \}$, where $\ker(R)$ is the null space or the kernel of the resolvent $R$.  Therefore, we can write $R(z)=(H-z)^{-1}$.
As a consequence of this result and the property of the resolvent $R(z)^{*}=R(\bar{z})$ from its explicit expression (\ref{resolvent}) and the symmetry property of the heat kernel, it is easy to see that $H$ is self-adjoint ($H^* - \bar{z}=(H-z)^*=(R^{-1}(z))^*=(R(z)^{*})^{-1}=R(\bar{z})^{-1}=H-\bar{z}$, where $*$ denotes the adjoint).
Then, it can be shown that the sequence of $H_{\epsilon}$ operators  converges to $H$ in the strong resolvent sense \cite{Albeverio 2}.

Unfortunately, Fourier transform on a general Riemannian manifold is absent. Nevertheless, using the Corollary 9.5 in \cite{Pazy}, it is possible to prove that there exists a densely defined closed operator $H$ associated with the resolvent (\ref{resolvent}) \cite{caglar}.  For convenience of the reader, we give the statement of that corollary:
Let
$\Lambda$ be an unbounded subset of $\mathbb{C}$. Then, $R(z)$ associated with the above resolvent (\ref{resolvent}) is a
pseudo resolvent on $\Lambda$. Moreover, if there is a sequence $E_k$ such that  $|E_k| \rightarrow \infty$ as $k \rightarrow \infty$ (for instance, choose $E_k  = -k |E_0| \in
\Lambda$, where $E_0$ is chosen to be below the lower bound $E_*$ on the ground state
energy which has been found in \cite{point interactions on manifolds2})  and
$ \lim_{k \rightarrow \infty} - E_k R(E_k) \psi =\psi$
for all $\psi \in \mathcal{H}$, then $R(z)$ is the resolvent of a
unique densely defined closed operator $H$. Then, self-adjointness of $H$ follows immediately from symmetry property as we have shown in the above paragraph. 
Hence, the sequence of  the operator $H_\epsilon$ converges to the operator $H$ in the strong resolvent sense{\color{red},} and this completes the proof. Actually, the existence of the self-adjoint Hamiltonian operator can also be proved using Trotter-Kato theorem \cite{Reed Simon V1} and this is discussed for this model and its many-body version on flat spaces in \cite{dimock}.

\section*{Appendix B: A Proof of the Proposition 1} \label{proof of prop1}

\setcounter{equation}{0}

\renewcommand{\theequation}{B-\arabic{equation}}

It is well known \cite{chavel} that the spectrum of the Laplacian on compact connected Riemannian manifolds without boundary only consists of the point part, i.e., $\sigma(-\triangle_g)=\{0=\sigma_0 \leq \sigma_1 \leq \ldots \}$, with $\sigma_l$ tending to infinity as $l \rightarrow \infty$ and each eigenvalue has finite multiplicity.  In order to show that the Hamiltonian $H$ is a compact perturbation of the free Hamiltonian, we first note that if $(H-z)^{-1} -(H_0 -z)^{-1}$ is compact for some $z \in \rho(H) \cap \rho(H_0)$, where $\rho$ denotes for the resolvent set, it holds for all $z\in \rho(H) \cap \rho(H_0)$ by  Lemma 4 in chapter XIII.4 \cite{Reed Simon V4}. Hence it suffices to prove it for a particular $z$. For that reason, let us choose $z=-E+ i \epsilon$, where $E$ is real and sufficiently large  positive and $\epsilon >0$. Then, compute $\Tr[(H-z)^{-1} -(H_0 -z)^{-1}]$ for any orthonormal basis $\{\phi_n\}$ in $L^2(M)$
\beqs \begin{split}
 \sum_n  \sum_{i,j=1}^{N} \left(\int_{M} \overline{\phi_{n}(x)} R_0(x,a_i|-E+i \epsilon) \; d_{g}^{D}x
\right) & [\Phi^{-1}(-E+i \epsilon)]_{ij} \\[2ex] & \times \, \left(\int_{M} \phi_{n}(y) R_0(y,a_j|-E+i \epsilon) \; d_{g}^{D}y
\right) \;. \end{split} 
\eeqs
Using the Cauchy-Schwarz inequality and the fact that $R_0(x,a_i|z)=\int_{0}^{\infty} \! K_{t_1}(x,a_i) \, e^{z t_1} \; d t_1$ (similarly for $R_0(y,a_j|z)$) for $\Re(z)<0$,  the above term is less than or equal to
\beqs \begin{split} \label{compactness of second term}
 N^2 \; \max_{1 \leq i,j \leq N} \bigg[ \left(\int_{0}^{\infty} \! K_{u_1}(a_i,a_i) \,e^{-u_1 E}\, {\sin u_1 \epsilon \over \epsilon} \; du_1 
\right)^{1/2} & | [\Phi^{-1}(-E+i \epsilon)]_{ij} |  \\[2ex] &  \times \, \left(\int_{0}^{\infty} \! K_{u_2}(a_j,a_j) \,e^{-u_2 E}\, {\sin u_2 \epsilon \over \epsilon} \; du_2
\right)^{1/2} \bigg] \;,
\end{split}
\eeqs
where we have used the semi-group property of the heat kernels and made change of variables $t_1+t_2=u$, $t_1-t_2=v$. One can then easily show that the above integrals are finite due to the upper bound of the heat kernel (\ref{heatkernel bound2}). 
Let us now recall the following fact (Corollary 5.6.13 in \cite{Horn}): Let $A$ be $N \times N$ matrix, and let $\eta >0$ be given. Then, there is a constant $C=C(A,\eta)$ such that $|(A^k)_{ij}| \leq C(\rho(A)+\eta)^k$ for all $k=1,2,\ldots$ and all $i,j=1,2,\ldots, N$, where $\rho(A)$ is the spectral radius of the matrix $A$. Let $A=\Phi^{-1}$ and $k=1$, then $|[\Phi^{-1}(-E+i \epsilon)]_{ij}| \leq C (\rho(\Phi^{-1}) +\eta) \leq C(||\Phi^{-1}(-E+i\epsilon)|| +\eta)$. Let $\Phi=\mathbb{D}-\mathbb{K}$, where $\mathbb{D}$ is the diagonal part of the matrix $\Phi$. Then, $\Phi^{-1}= (1-\mathbb{D}^{-1}\mathbb{K})^{-1} \; \mathbb{D}^{-1}$. Therefore, the principal matrix $\Phi$ is invertible if and only if $(1-\mathbb{D}^{-1}\mathbb{K})$ has an inverse. The matrix $(1-\mathbb{D}^{-1}\mathbb{K})$ is invertible if $||\mathbb{D}^{-1}\mathbb{K}||<1$. Then, the inverse of $\Phi$ can be written as a geometric series $\Phi^{-1}=(1+(\mathbb{D}^{-1}\mathbb{K})+(\mathbb{D}^{-1}\mathbb{K})^2 + \ldots)\; \mathbb{D}^{-1}$  from which we get $||\Phi^{-1}||\leq {1 \over 1-||\mathbb{D}^{-1}\mathbb{K}||} \; ||\mathbb{D}^{-1}||$. If we choose $E$ sufficiently large that $||\mathbb{D}^{-1}\mathbb{K}||=1/2$ we find $||\Phi^{-1}(-E+i\epsilon)|| \leq 2 ||\mathbb{D}^{-1}(-E+i \epsilon)||$, which is bounded from above by Lemma \ref{CheegerYau}. Hence, we show that the operator $(H-z)^{-1} -(H_0 -z)^{-1}$ is trace class so it is compact for sufficiently large values of $E$ (hence for all $z \in \rho(H) \cap \rho(H_0)$). Since there are points of $\rho(H_0)$ in both upper and lower half-planes,
$\sigma_{ess}(H)=\sigma_{ess}(H_0)=\varnothing$ due to the Weyl's essential spectrum theorem \cite{Reed Simon V4}.

\section*{Appendix C: A Proof of the Theorem 2} \label{proof of spectrumtheorem}

\setcounter{equation}{0}

\renewcommand{\theequation}{C-\arabic{equation}}

Since the point spectrum of the self-adjoint operator $H$ is given by the set of
real numbers such that the resolvent of that operator does not exist, the resolvent formula (\ref{resolvent}) shows that its negative poles can only occur when the principal matrix is noninvertible, i.e., the solution to the characteristic equation $\det \Phi(\nu) =0$ contributes to the negative part of the point
spectrum of $H$, whereas the free resolvent has positive real poles. Since the positive part of the point spectrum is due to the free Hamiltonian, we interpret the negative part of the point spectrum as bound states.

Let $z=-\nu_{k}^2$ be one of the negative isolated poles of the resolvent. Then, the orthogonal projection onto the subspace
$\ker(H+\nu_{k}^2)$ is given by the Riesz integral representation for $H$ 
\cite{Reed Simon V4}:
\beqs \mathbb{P}_{k} =
- {1\over 2\pi i} \oint_{\Gamma_k}
\mathrm{d} z \;  R(z) \;, \eeqs
where $\Gamma_k$ is an admissible contour enclosing only the isolated
pole $-\nu_k^2$. 
Since the principal matrix is symmetric (self-adjoint) on the real axis, we can write the spectral decomposition of it. Moreover, we can use the fact that $\Phi$ is holomorphic so that there exists holomorphic family of
projection operators on the complex plane \cite{Kato}. Hence, the spectral resolution of
the inverse principal matrix exists and given by
$\Phi^{-1}_{ij}(z)=\sum_{n} {1 \over \omega_n(z)} \mathbb{P}_n(z)_{ij}$,
where $\mathbb{P}_n(z)_{ij}=\bar{A}_i^{n}(z) A_j^n(z)$, and $A^n_i(z)$ is the
normalized eigenvector corresponding to the eigenvalue
$\omega_n(z)$. 
Above contour integral can be calculated from residue theorem and Feynman-Hellmann theorem \cite{Feynman Hellmann1, Feynman Hellmann2} (actually this theorem is also stated without referring Feynman-Hellmann in \cite{Kato}), from which  we can find the wave function (\ref{wavefunction heat delta}) associated with the pole $-\nu_{k}^2$. After a tedious but straightforward computation, we find that the eigenvalues flow according to ${d \omega_n \over d \nu} >0$ 
as a consequence of Feynman-Hellmann theorem and positivity of the heat kernel (see the details in \cite{point interactions on manifolds2}).
Since bound states are obtained from the zeros of the
eigenvalues of the principal matrix, namely,
$\omega_n(-\nu_k^2)=0$,  there is a unique solution for each
$\omega_n(-\nu^2)$ due to its  monotonic behavior. Hence, each eigenvalue $\omega_n$ has at most one zero in $(0,\infty)$.
This implies that there can be at most $N$ zeroes of the eigenvalues, say  $\nu_1, \ldots , \nu_N$, i.e., there can be at most $N$ negative eigenvalues of $H$.

Let $E=\nu_{k}^2$ be an eigenvalue of $H$. Suppose that this eigenvalue does not coincide with the poles of the free resolvent. Then, from the explicit expression of (\ref{wavefunction heat delta}), the wave function associated with this positive isolated pole can not be in $L^2(M)$ unless it is identically zero. This proves the absence of nonnegative eigenvalues coming from the principal matrix $\Phi$.

Nondegeneracy of the ground state and the positivity of its eigenfunction follows from the Perron-Frobenius theorem
\cite{point interactions on manifolds2}.

\section*{Appendix D: A Proof of Proposition \ref{numberof bound states flatspace}}
 \label{proof theorem numberofbs}

 \setcounter{equation}{0}

\renewcommand{\theequation}{D-\arabic{equation}}

Let us first prove the two dimensional case. Using the well-known explicit expression of the heat kernel in $\mathbb{R}^2$, the principal matrix $\Phi$ restricted to the negative real axis $z=-\nu^2$ is 
\beqs\Phi_{ij}(-\nu^2) =  {1 \over 2 \pi} \; \log(\nu/\mu_i) \delta_{ij} -
(1-\delta_{ij})  {1 \over 2 \pi} \; K_0\left( \nu d_{ij} \right) \;,
\label{renormalized Phi delta 2}
\eeqs
where $K_0$ is the modified Bessel function of the third kind, or Macdonald's function \cite{Lebedev} and $d_{ij}=|a_i-a_j|$.

The characteristic equation for $N=2$ yields
$ \log(\nu/\mu_1) \log(\nu/\mu_2) = K_0^2\left(\nu d \right)$.
Let $x=\nu d$, $\alpha_1={1
\over \mu_1 d}$ and $\alpha_2={1
\over \mu_2 d}$, so that it becomes
\beqs \log(\alpha_1 x) \log(\alpha_2 x) = K_0^2(x) \;. \label{2 delta transcendental 2D dimless} \eeqs
in the dimensionless variables. Although the roots of the above transcendental
equation (\ref{2 delta transcendental 2D dimless}) can not be analytically found, we can at
least determine how many roots (bound states) we have and what sufficient conditions must be met for the maximum number of roots. The left hand side of (\ref{2 delta transcendental 2D dimless})
is a positive decreasing function when $0<x<{1\over \alpha_1}$ and a
positive increasing one when $x>{1 \over \alpha_2}$, whereas it has one zero at $x={1 \over \sqrt{\alpha_1 \alpha_2}}$. Hence, $\log(\alpha_1 x) \log(\alpha_2 x)$ has a local
minimum at $x= {1 \over \sqrt{\alpha_1 \alpha_2}}$.
No matter how $\alpha_1$ and $\alpha_2$ are chosen, we expect that there
is at least one root because the function on the left-had side of (\ref{2 delta transcendental 2D dimless})
eventually intersects the
monotonically positive decreasing function on the right-hand side of it ($K_{0}^{2}(x) \sim {\pi \over 2 x} e^{-2x} \left(1 + O(1/x)
\right) $
as $x\rightarrow \infty$). This tells us that there exists at least one bound state. 
We may have a second root if we
impose the condition that $\log^2(\alpha x)$ is
able to exceed the function $K_0^2(x)$ near $x=0$. Therefore it is
necessary to impose 
(\ref{2 delta transcendental 2D dimless}) for $x<1/\alpha_1$ in order to get a second bound
state. Using the asymptotic expansion of the Bessel function $K_0(x)$ \cite{Lebedev}
\beqs
K_0(x) \sim - \log(x/2) \; \hspace{2cm} \mathrm{as} \; \; x \rightarrow 0 \;,
\eeqs 
and imposing (\ref{2 delta transcendental 2D dimless}) near $x=0$, we obtain the claimed condition $d > {2 \over \sqrt{\mu_1 \mu_2}}$.

The principal matrix restricted to real negative energies in $\mathbb{R}^3$ is given by
\beqs
\Phi_{ij}(-\nu^2)=
\frac{1}{4 \pi }(\nu-\mu_{i}) \delta_{ij}
-(1-\delta_{ij}) \frac{1}{4\pi }\frac{e^{-\nu d_{ij}}}{d_{ij}}
\;. \label{principle matrix for three renormalized}
\eeqs
Then, the characteristic equation for $N=2$ leads to 
\beqs
(\nu-\mu_{1})(\nu-\mu_{2})=\frac{1}{d^{2}}e^{-2d\nu} \;. \label{bound state equation 3D}
\eeqs
Introducing the dimensionless variables $x=\nu d$, $\alpha_1=\mu_1 d$, and $\alpha_2=\mu_2 d$, the above equation simply becomes $(x-\alpha_1) (x-\alpha_2)=e^{-2x}$. We can assume that $\alpha_1 < \alpha_2$ without loss of generality. The left-hand side of this equation is a parabola whose zeroes are $\alpha_1$ and $\alpha_2$. If the exponential function on the right-hand side at $x=0$ is less than the value of the function at $x=0$ on the left-hand side, then there are two roots (bound states) of this equation. This means that if  
we impose $ LHS(0) > RHS(0)$, we arrive at the claimed condition 
$d>{1 \over \sqrt{\mu_1 \mu_2}}$. 

Actually, we can also obtain the above conditions by following the same arguments developed in Section \ref{Main Results} (the infimum of the functions are very easy to find).



\begin{thebibliography}{99}







\bibitem{Albeverio 1} Albeverio S, Gesztesy F, Høegh-Krohn R and Holden H (2004) Solvable Models in Quantum Mechanics, 2nd ed. (AMS
Chelsea, RI)



\bibitem{Albeverio 2} Albeverio S and Kurasov P (2000) Singular Perturbations of Differential Operators
Solvable Schr\"{o}dinger-type Operators (Cambridge
University Press, Cambridge)



\bibitem{Demkov} Demkov Yu N and Ostrovskii V N (1988) Zero-range Potentials and Their
Applications in Atomic Physics (Plenum Press, New York)


\bibitem{Berezin Faddeev} Berezin F A and Faddeev L D (1961) Soviet Math. Dokl. \textbf{2} 372-375


\bibitem{Rudnick torus} Rudnick Z and  Uebersch\"{a}r H (2012) Comm. Math. Phys. \textbf{316}, 763-782 
 

\bibitem{Exner strip} Exner P,  Gawlista R, Seba P and Tater M (1996) Ann. Physics \textbf{252} Issue 1 133–179 
  
  
\bibitem{BGP} Br\"{u}ning J, Geyler V and Pankrashkin K (2008) Rev. Math. Phys. 20 170


\bibitem{point interactions on manifolds1} Altunkaynak B, Erman F and Turgut O T (2006) 
J. Math. Phys. \textbf{47} 082110

\bibitem{point interactions on manifolds2} Erman F and  Turgut  O T (2010) J. Phys. A: Math. Theor. \textbf{43} 335204


\bibitem{caglar} Dogan C,  Erman F and  Turgut O T (2012) J. Math. Phys. \textbf{53}, 043511



\bibitem{Reed Simon V4}  Reed M and  Simon B (1978) Methods of Modern Mathematical Physics, \textbf{IV} (Academic Press, NewYork)



\bibitem{Albeverio Nizhnik 1} Albeverio S and  Nizhnik L (2003) Methods Funct. Anal. Topology \textbf{9} (4) 273-286 


\bibitem{Albeverio Nizhnik 2} Albeverio S and Nizhnik L (2003) Lett. Math. Phys. \textbf{65} 27-35


\bibitem{Ogurisu 1} Ogurisu O (2008) Lett. Math. Phys. \textbf{85} 129


\bibitem{Ogurisu 2} Ogurisu O (2010) Methods Funct. Anal. Topology
\textbf{16} 42-50

\bibitem{Goloshchapova Oridoroga 1} N. I. Goloshchapova, L. L. Oridoroga, The one-dimensional Schr\"{o}dinger operator with point $\delta$- and $\delta'$- interactions, Math. Notes \textbf{84} no. 1, 125-129 (2008)


\bibitem{Goloshchapova Oridoroga 2} Goloshchapova N I and  Oridoroga L L (2010)  Integral Equations Operator Theory \textbf{67} no. 1, 1-14



\bibitem{Ogurisu 3} Ogurisu O (2010)  Methods of Funct. Anal. Topology,
\textbf{16}  4, 383 - 392

\bibitem{Jackiw} Jackiw R (1991) Delta-Function Potentials in Two- and Three-Dimensional Quantum Mechanics, M. A. B. Beg Memorial Volume (World Scientific, Singapore)



\bibitem{AlHashimi}   Al-Hashimi M H,  Shalaby A M and  Wiese U –J (2014) Phys. Rev. D \textbf{89} Issue: 12, 125023 


\bibitem{DeWitt}  De Witt B (1957) Rev. Modern Phys.  \textbf{29} 377






\bibitem{Gilkey}  Gilkey P B (1984) Invariance Theory, the Heat Equation, and the Atiyah-Singer Index Theorem (DE: Publish or Perish, Wilmington)




\bibitem{Wang}  Wang J (1997) Pasific
Journal of Mathematics, \textbf{178} No 2 377


\bibitem{Grigoryan Heat Book} Grigor’yan A (2009) Heat Kernel and Analysis on Manifolds, AMS/IP Studies in Advanced Mathematics Vol. 47, edited by S.-T.
Yau, (American Mathematical Society, Rhode Island)

\bibitem{Grigoryan 2ndbook} Grigor’yan A (1999) Spectral Theory and Geometry, London Mathematical Society Lecture Notes Vol. 273, edited by E. B.
Davies and Y. Safarov (Cambridge University Press, Cambridge), pp. 140–225


\bibitem{CheegerYau} Cheeger  J and Yau S-T (1981) Comm. Pure Appl. Math. \textbf{34}, 465 - 480 


\bibitem{Olver}  Olver F W J (1974) Asymptotics and Special Functions (Academic Press, New York)
 
 
\bibitem{Copson}  Copson E T (1935) An Introduction to the Theory of Functions of a Complex Variable (Oxford University Press, London)
 
 
\bibitem{Lebedev}  Lebedev N N (1965) Special Functions and Their
Applications (Printice Hall, NJ Englewood Cliffs) 
 
\bibitem{Kato} Kato T (1995) Perturbation Theory for Linear Operators,
Classics in Mathematics (Springer-Verlag, corrected
printing of the second edition, Berlin) 


\bibitem{Donnelly0} Donnelly H (1990) Math. Z. \textbf{203}, 301-308 

\bibitem{CL} Charalambous N and Lu Z (2014) Math. Ann. 359:211–238 


\bibitem{Donnelly 2} Donnelly H (1981) Topology \textbf{20}  1–14 


\bibitem{Feynman Hellmann1}  Feynman R P (1939) Phys. Rev. \textbf{56} 340


\bibitem{Feynman Hellmann2} Hellmann H G A (1933) Z. Phys. \textbf{85} 180–190


\bibitem{cohenbook} Cohen-Tannoudji C,  Diu B and Laloe F (2006)
Quantum Mechanics, Vol. 1 (Wiley-Interscience, New York)


\bibitem{Davies Mandouvalos}  Davies E B and Mandouvalos N (1988) Proc. London Math. Soc. \textbf{3} 57, 182-208

\bibitem{Gradshteyn}  Gradshteyn I S and  Ryzhik I M (2007) Table of Integrals, Series, and Products, Seventh Edition (Academic Press)

\bibitem{Horn} Horn R A  and  Johnson C R (1992) Matrix Analysis
(Cambridge University Press, Cambridge)

\bibitem{Caglar1} Dogan  C and  Turgut O T (2000) J. Math. Phys. \textbf{51} 8 082305

\bibitem{Pazy}  Pazy A (1983) Semigroups of Linear Operators and
Applications to Partial Differential Equations
(Springer-Verlag, New York)


\bibitem{Reed Simon V1}   Reed M and Simon B (1980) Methods of Modern Mathematical Physics, \textbf{I} (Academic Press, Revised and Enlarged edition, NewYork)


\bibitem{dimock} Dimock  J and Rajeev S G (2004) J. Phys. A: Mathematical and General, \textbf{37}, Number 39


\bibitem{chavel} Chavel I (1984) Eigenvalues in Riemannian
Geometry, Pure and Applied Mathematics, Vol. 115 (Academic Press,
Orlando)
















\end{thebibliography}
\end{document}